\documentclass[]{article}

\usepackage{tikz, pgfplots, pgfplotstable}
\usepackage{exmath}  
\usepackage{xargs}
\usepackage[ruled, vlined]{algorithm2e}
\usepackage{comment}
\usepackage[nocompress]{cite}
\usepackage{amsfonts, amsmath, amssymb, amsthm}
\usepackage{hyperref, cleveref}
\usepackage{enumitem}
\usepackage{multirow}
\usepackage[]{geometry}

\newtheorem{theorem}{Theorem}[section]
\newtheorem{lemma}[theorem]{Lemma}

\newtheorem{proposition}[theorem]{Proposition}
\newtheorem{remark}[theorem]{Remark}
\newtheorem{assumption}[theorem]{Assumption}

\crefname{assumption}{assumption}{assumptions}
\newcommand\antdd{\dl{^{\textnormal{ant}}}\dl[k]{X}\p{\Yantf{\ell}{0}, \Yantf{\ell}{1}; \Yantc{\ell-1}}}

\usepackage{graphicx}
\usepackage{xcolor}
\hypersetup{  colorlinks,
    linkcolor={blue!50!black},
    citecolor={green!50!black},
    urlcolor={red!80!black}
}

\RequirePackage{xargs}

\newcommand\Yt[1][1]{Y_{#1}}  \newcommandx\Ytell[2][1=1,2=\ell]{\Yt[#1,#2]}  

\DeclareMathOperator{\Heavsymb}{\mathnormal{H}}

\xParseDeclareExpectation{\Ern}{\mathsf{E}^{\mathsf{Q}}}[]
\xParseDeclareExpectation{\Ephys}{\mathsf{E}^{\mathsf{P}}}[]
\xParseDeclareExpectation{\Heav}{\Heavsymb}()

\newcommand{\dl}[1][\ell]{{\Delta_{#1}}} \newcommandx{\dldl}[3][1=\ell,2=k,3=g]{\dl[#1]{\dl[#2]{#3}} }

\newcommand\tol{\varepsilon}  

\newcommand\costmlmc[1][\tol]{\textnormal{Cost}_{\textnormal{MLMC}}\p*{#1}}

\newcommand\nestE[1][0]{U_{#1}}  \newcommandx\appE[2][1=\ell, 2=1]{\widehat{U}_{#2, #1}}  \newcommandx\appEml[2][1=\ell, 2=1]{\widehat{U}_{#2, #1}^{\p{\textnormal{ML}}}}  \newcommandx\antEc[3][1=0,2=\ell-1,3=1]{\widehat{U}_{{#3}, {#2}}^{\p{\textnormal{c}, #1}}}  \newcommandx\antEf[2][1=\ell,2=1]{\widehat{U}_{{#2}, {#1}}^{\p{\textnormal{f}}}}  \newcommandx\antEcml[3][1=0,2=\ell-1,3=1]{\widehat{U}_{{#3}, {#2}}^{\p{\textnormal{ML}, \textnormal{c}, #1}}}  \newcommandx\antEfml[2][1=\ell,2=1]{\widehat{U}_{{#2}, {#1}}^{\p{\textnormal{ML}, \textnormal{f}}}}  \newcommandx\antEcmly[3][1=0,2=\ell-1,3=1]{\widehat{U}_{{#3}, {#2}}^{\p{\textnormal{ML}, \textnormal{c}, #1}}}  \newcommandx\antEfmly[2][1=\ell,2=1]{\widehat{U}_{{#2}, {#1}}^{\p{\textnormal{ML}, \textnormal{f}}}}  \newcommandx\antEfmlun[2][1=\ell,2=1]{\overline{U}_{{#2}, {#1}}^{\p{\textnormal{ML}, f}}}  \newcommandx\Ml[2][1=0,2=1]{N_{#1}}  \newcommandx\Mlk[3][1=\ell,2=k,3=1]{\Ml[#1,#2][#3]}  

\newcommand\bq[1][q]{\beta_{#1}}  \newcommand\bqy[1][q]{\theta_{#1}}  \newcommand\Xl[1][k]{X_{#1}}

\newcommand\tpayoff[1][Y]{\pi\p{#1}}

\newcommand\Yantf[2]{Y_{#1}^{\textnormal{f}, #2}}
\newcommand\Yantc[1]{Y_{#1}^{\textnormal{c}}}

\newcommand\dlant[1][\ell]{\Delta_{#1}^\textnormal{(ant)}}
\newcommand\dlantml[1][\ell]{\Delta_{#1}^{\p{\textnormal{ML}}}}
\newcommand\dlantmly[1][\ell]{\Delta_{#1}^{\p{\textnormal{ML}, Y}}}
\newcommand\dlantpath[1][\ell]{\Delta_{#1}^{\p{\textnormal{ant}, Y}}}
\newcommandx\deltemp[2][1=\ell,2=t_1]{\Delta_{#1, #2}}
\newcommandx\dlj[2][1 = \ell, 2 = j]{\Delta_{#2, #1}^{\p{\textnormal{rec}, T}}}  \newcommandx\dljpath[2][1=j,2=\ell]{\Delta_{#1, #2}^{\p{\textnormal{path}, T}}}

\newcommandx\sgma[2][2=\ell]{\sigma\p{#1}}
\newcommandx\sgmaq[2][2=\ell]{\sigma^q\p{#1}}
\newcommandx\sgmaj[4][2=\ell,3=j,4=q]{\sigma_{\p{#3}}^{#4}\p{#1}}
\newcommandx\sgmahat[2][2=\ell]{\hat{\sigma}^q\p{#1}} \RequirePackage{pgfplots, tikz, pgfplotstable}

\pgfplotsset{compat=newest}

\pgfplotsset{width = 5.5cm,
			 every axis plot post/.append style = {line width = 1pt,
		 		every mark/.append style = {solid}}}
	 		
\newcommand\opac{0.6}  %
 
\begin{document}
	\title{\vspace{-2cm}Nested Multilevel Monte Carlo with Biased and Antithetic Sampling}
\author{Abdul-Lateef Haji-Ali\thanks{Maxwell Institute for Mathematical Sciences, Department of Actuarial Mathematics and Statistics, Heriot-Watt University, Edinburgh, UK, EH14 4AS. (\href{mailto::a.hajiali@hw.ac.uk}{a.hajiali@hw.ac.uk}).}
	\and Jonathan Spence\thanks{Corresponding author. Maxwell Institute for Mathematical Sciences, Department of Actuarial Mathematics and Statistics, Heriot-Watt University, Edinburgh, UK, EH14 4AS. (\href{mailto::jws5@hw.ac.uk}{jws5@hw.ac.uk}).}}
	
\date{}
\maketitle
\begin{abstract}
\textbf{AMS Subject Classication: } 65C05, 62P05


We consider the problem of estimating a nested structure of two expectations taking the form \(U_0 = \mathsf{E}[\max\{U_1(Y), \pi(Y)\}]\), where \(U_1(Y) = \mathsf{E}[X\ |\  Y]\). Terms of this form arise in financial risk estimation and option pricing. When \(U_1(Y)\) requires approximation, but exact samples of  \(X\) and \(Y\) are available, an antithetic multilevel Monte Carlo (MLMC) approach has been well-studied in the literature. Under general conditions, the antithetic MLMC estimator obtains a root mean squared error \(\varepsilon\) with order \(\varepsilon^{-2}\) cost. If, additionally, \(X\) and \(Y\) require approximate sampling, careful balancing of the various aspects of approximation is required to avoid a significant computational burden. Under strong convergence criteria on approximations to  \(X\) and \(Y\), randomised multilevel Monte Carlo techniques can be used to construct unbiased Monte Carlo estimates of \(U_1\), which can be paired with an antithetic MLMC estimate of \(U_0\) to recover order \(\varepsilon^{-2}\) computational cost.  In this work, we instead consider biased multilevel approximations of \(U_1(Y)\), which require less strict assumptions on the approximate samples of \(X\). Extensions to the method consider an approximate and antithetic sampling of \(Y\). Analysis shows the resulting estimator has order \(\varepsilon^{-2}\) asymptotic cost under the conditions required by randomised MLMC and order \(\varepsilon^{-2}|\log\varepsilon|^3\) cost under more general assumptions.
\end{abstract} 	\section{Introduction}
	This paper considers the nested use of multilevel Monte Carlo (MLMC) techniques to approximate {two nested} expectations requiring recursive numerical approximation. We focus on {the problem}
\begin{equation}
\label{eqn:nestexp2}
\begin{aligned}
	\nestE[0] &= \E*{\max\br*{\nestE[1]\p{Y}, \tpayoff}}\\
	\nestE[1]\p{Y} &= \E*{X\given Y},
\end{aligned}
\end{equation}
{defined for some \(\rset\)-valued random variable \(X\) and \(\rset^d\)-valued random variable \(Y\), and with payoff function \(\pi:\rset^d{\to} \rset\). Of particular interest in this work is the case where \(\nestE[1]\p{Y}\) requires approximation, and neither \(X\) nor \(Y\) can be sampled directly.} In {financial applications}, \(Y\) can represent a random intermediate state of an underlying {asset} {whereas \(X\) can represent future profits or losses after the intermediate state \(Y\).} {Nested expectations of this form are of interest in financial risk estimation \cite{GilesHajiAli:2018,hajiali2021adaptive,GilesHajiAliSpence:CVA,Gordy:2010}. For example, the expected shortfall can be expressed in the form \eqref{eqn:nestexp2} with no intermediate payoff (\(\tpayoff = 0\)) and where \(X\) represents losses above the value-at-risk, conditioned on a risk scenario \(Y\) \cite{GilesHajiAli:2018,rtu:cvar,cfl:2023}. An additional application is in non-European option pricing \cite{zhou2022}, where \(Y\) again represents asset values at some intermediate time. {Given \(Y\), it is possible to either receive some immediate payoff \(\tpayoff\) or to wait until a future time and receive a random payoff \(X\), depending on the intermediate state \(Y\).}   Additional applications include credit risk estimation \cite{GilesHajiAliSpence:CVA} and optimal control \cite{zhou2022,gg19}.}   \\

When \(X\) and \(Y\) can be sampled directly, an MLMC estimate using antithetic Monte Carlo {estimates} of \(\nestE[1]\p{Y}\) is analysed in \cite{GilesHajiAli:2018,Bourgey20,BujokK2015}. Under general conditions, this approach approximates \(\nestE[0]\) with root mean square error \(\tol\) at cost of order \(\tol^{-2}\). The {authors in} \cite{GilesHajiAliSpence:CVA} use{d} a similar approach to consider the case where \(X\) and \(Y\) are given by a Brownian SDE to be approximated using Milstein discretisation. There, unbiased estimates of \(X\) are sampled using randomised MLMC techniques \cite{Rhee2015_Unbiased_estimation}. The unbiased estimates are then used within an antithetic Monte Carlo difference coupled with a multilevel approximation to \(Y\). Using randomised MLMC to sample unbiased estimates of \(X\) simplifies the resulting estimator by removing an extra layer of bias but requires strict convergence assumptions on the approximation to \(X\). In \cite{zhou2022,syed2023optimal}, the authors consider a recursion of the problem \eqref{eqn:nestexp2}, taking \(X = \max\br{\nestE[2]\p{Y^\prime}, \pi\p{Y^\prime}}\), with \(\nestE[2]\p{Y^\prime} = \E{X^\prime\given Y^\prime}\) for {an \(\rset\)-valued random variable \(X^\prime\) and \(\rset^d\)-valued random variable \(Y^\prime\), depending on \(Y\)}. Assuming exact samples of \(X^\prime, Y, Y^\prime\) are available, the authors randomise the antithetic MLMC approach in \cite{Bourgey20,GilesHajiAli:2018,BujokK2015}  to construct unbiased estimates of \(X\). The method is then applied iteratively to construct an unbiased estimate for an arbitrary number of nested expectations of this form.\\

The key contributions in {the current work} are as follows:
\begin{enumerate}
	\item {\Cref{sec:twostageexactmc} provides an overview of Monte Carlo and (antithetic) multilevel Monte Carlo estimation of \eqref{eqn:nestexp2} when \(X\) and \(Y\) can be sampled directly. \Cref{lem:ant_general,lem:ant_general_smooth} contain a general convergence result for antithetic multilevel Monte Carlo differences.}
	\item Analysis of a nested (biased) MLMC approximation for the two-stage problem \eqref{eqn:nestexp2} using approximate samples of \(X\) and \(Y\)  is presented in \Cref{sec:twostageexact}. The method permits biased approximations to \(\nestE[1]\p{Y}\), which permit biased approximations of \(X\).  {The analysis includes convergence of arbitrary moments of the multilevel estimates. Convergence of higher moments is useful, for example, if applying the approach recursively to deal with more layers of nested expectations as in \cite{zhou2022,syed2023optimal}, see \Cref{rem:two_stage_prop_recursive}; or when paired with more sophisticated estimators such as adaptive MLMC \cite{GilesHajiAli:2018,hajiali2021adaptive,Elfverson:2016selectiverefinement} or multilevel path branching \cite{gh22b}.}
	\item \Cref{sec:ant_path_2stage} extends the above approach to include antithetic simulation of \(Y\), resulting in a doubly antithetic estimator for \(\nestE[0]\). This approach helps improve convergence rates when \(Y\) is the solution of an SDE for which Milstein simulation is infeasible due to the computation of L\'evy areas \cite{giles14antmilstein,cc80}.
	\item Numerical experiments supporting the theoretical results are presented in \Cref{sec:numerics}.
\end{enumerate} 	
	\section{Exact Nested Sampling}
\label{sec:twostageexactmc}

			We begin under the assumption that exact samples of \(X\) and \(Y\) are available, but the inner expectation \(\nestE[1]\p{Y}\) requires approximation. One approach to approximate \(\nestE[0]\) is using a Monte Carlo average over \(Y\), with an inner nested Monte Carlo approximation of \(\nestE[1]\p{Y}\) given \(Y\). That is, for \(\Ml[L] = \Ml[0]2^L\), let
\[
	\begin{aligned}
	\nestE[1]\p{Y} \approx \appE[L]\p{Y} \defeq \frac{1}{\Ml[L]}\sum_{n=1}^{\Ml[L]} X^{(n)}\p{Y},
	\end{aligned}
\] 
{where \(\br{X^{(n)}\p{y}}_{n=1}^{\Ml[L]}\) are conditionally independent samples of \(X\) given \(Y = y\)}. For \(M\) independent samples of \(Y\) denoted by \(\br{Y^{\p{m}}}_{m=1}^M\), a (biased) nested Monte Carlo estimate of \(\nestE[0]\) is then given by
\[
\begin{aligned}
\nestE[0] &\approx \E*{\max\br*{\appE[L]\p{Y}, \tpayoff}}\\
&\approx \frac{1}{M}\sum_{m=1}^{M} \max\br*{\appE[L]\p{Y^{(m)}},\ \tpayoff[Y^{\p{m}}]}\\
&= \frac{1}{M}\sum_{m=1}^{M} \max\br*{\frac{1}{\Ml[L]}\sum_{n=1}^{\Ml[L]} X^{(m,n)}\p{Y^{(m)}} ,\ \tpayoff[Y^{\p{m}}]}.
\end{aligned}
\]
{
Under a Lipschitz condition on the distribution of \(\nestE[1]\p{Y}\) and assuming \(\E{\abs{X}^q}<\infty\) for some \(q>2\), the bias \(\abs{\E{\max\br{\nestE[1]\p{Y}, \tpayoff}} - \E{\max\br{\appE[L]\p{Y}, \tpayoff}}}\) induced by approximating \(\nestE[1]\) by \(\appE[L]\) is of order \(\Ml[L]^{-q/\p{q+1}}\) \cite[Proposition 4]{Bourgey20},\cite[Proposition 4]{BujokK2015}. To achieve a root mean square error \(\tol\), one should take \(M = \Order{\tol^{-2}}\) and \(L = \Order{\log\p{\tol^{-\p{q+1}/q}}}\) giving a total cost of \(M\Ml[L] = \Order{\tol^{-3 - 1/q}}\).}\\

To improve this cost, multilevel Monte Carlo (MLMC) techniques \cite{Giles2008MLMC} {approximate} \(\nestE[1]\p{Y}\) {at} all levels of accuracy \(0\le \ell\le L\). Defining the multilevel correction terms 
\begin{equation}
\label{eqn:dl_mcdiff}
\dl\p{Y} = 
\begin{cases}
\max\br*{\appE\p{Y}, \tpayoff} - \max\br*{\appE[\ell-1]\p{Y}, \tpayoff} & \text{if } \ell > 0\\
\max\br*{\appE[0]\p{Y}, \tpayoff}	&	\text{otherwise},
\end{cases}
\end{equation}
{where \(\appE\p{Y}\) and \(\appE[\ell-1]\p{Y}\) are conditionally independent given \(Y\)}. The MLMC estimator is given by
\begin{equation}
\label{eqn:mlmc}
\begin{aligned}
\nestE[0] \approx \E*{\max\br*{\appE[L]\p{Y}, \tpayoff}} &= \sum_{\ell=0}^L \E{\dl\p{Y}}\\
&\approx \sum_{\ell=0}^L \frac{1}{{M_\ell}}\sum_{m=1}^{{M_\ell}} \dl\p{Y^{(m)}},
\end{aligned}
\end{equation}
{where \(\br{\dl \p{Y^{(m)}}}_{m=1}^{{M_\ell}}\) are independent samples of \(\dl\p{Y}\).}
Since the map \(\p{u, y}\mapsto \max\br{u, y}\) is Lipschitz, the variance of the terms \(\dl\p{Y}\) is of order \(\E{\abs{\appE\p{Y} - \appE[\ell-1]\p{Y}}^2} = \Order{2^{-\ell}}\) assuming \(X\) has finite variance. {This approach requires fewer samples at {deep levels} to control the statistical error. As a result, MLMC complexity theory \cite{Cliffe:2011, Giles2008MLMC} provides values of \(L\) and \( \{M_\ell\}_{\ell=0}^L\) such that the cost of achieving root mean square error \(\tol\) is \(\Order{\tol^{-2}\p{\log\tol}^2}\).}\\

\subsection{Antithetic Monte Carlo Difference}
To remove the logarithmic term from the MLMC cost, we require \(\E{\abs{\appE\p{Y} - \appE[\ell-1]\p{Y}}^2}\) to decrease at a faster rate than order \(2^{-\ell}\). The related works \cite{GilesHajiAli:2018,Bourgey20,BujokK2015,gg19} achieve this result  by constructing an antithetic difference that better correlates the coarse and fine approximations in \(\dl\p{Y}\). The approach is based on the simple observation that whenever \(\p{x_0 - y}\p{x_1 - y} > 0\) we have 
\[
\max\br*{\frac{x_0+x_1}{2}, y} - \frac{1}{2}\p*{\max\br*{x_0, y} + \max\br*{x_1, y}} = 0.
\]
This fact leads to the following result, which we utilise throughout {the current work}:
\begin{lemma}
	\label{lem:ant_general}
	Let \(Z\) represent {an \(\rset\)-valued} random variable such that for some \(\bar \delta, \bar \rho > 0\),  we have
	\begin{equation}
		\label{eqn:cdf_bound_general}
		\prob{\abs{Z} < x} \le \bar\rho x,	
	\end{equation}
	for all \(x \le \bar\delta\). Consider  {an approximation \(Z_\ell\) of} \(Z\), satisfying the condition that for some \(q > 2\) we have
	\[
		\E*{\abs{Z_\ell - Z}^q} < \infty.
	\]
	Let	\(Z_\ell^{\p{0}}, Z_\ell^{\p{1}}\) be two samples of \(Z_\ell\).
	{For all \(2 \le p\le q\),  {there is \(b_0 > 0\) independent of \(\ell\) such that}
	\[
		\E*{\abs*{\max\br*{\frac{Z_\ell^{\p{0}} + Z_\ell^{\p{1}}}{2}, 0} - \frac{1}{2}\sum_{i=0}^1 \max\br*{Z_\ell^{\p{i}}, 0}}^p} \le b_0\E*{\abs*{Z_\ell - Z}^q}^{\p{p+1}/\p{q+1}}.
	\]}
\end{lemma}
{The proof of \Cref{lem:ant_general} uses the linear growth condition on the distribution of \(Z\) around 0 and the bound on \(\E{\abs{Z_\ell - Z}^q}\) to control the probability that 
\[
\max\br*{\frac{Z_\ell^{\p{0}} + Z_\ell^{\p{1}}}{2}, 0} - \frac{1}{2}\sum_{i=0}^1 \max\br*{Z_\ell^{\p{i}}, 0}\neq 0.
\]
}
The result follows immediately from a more general statement proven in \Cref{lem:ant_general_smooth}, which may be of interest by itself.\\

Motivated by \Cref{lem:ant_general}, consider the antithetic correction terms \(\dlant\) defined for \(\ell>0\) by
\begin{equation}
\label{eqn:antf}
\dlant \p{Y} \defeq \max\br*{\antEf\p{Y}, \tpayoff} - \frac{1}{2}\sum_{i=0}^1{\max\br*{\antEc[i]\p{Y}, \tpayoff}},
\end{equation}
where \(\antEc[0]\p{Y}\) and \(\antEc[1]\p{Y}\) are conditionally independent samples of  \(\appE[\ell - 1]\p{Y}\) given \(Y\), and
\[
\begin{aligned}
\antEf\p{Y} &= \frac{1}{2}\p*{\antEc[0]\p{Y} + \antEc[1]\p{Y}}.
\end{aligned}
\]
Importantly, by the linearity of expectation, it follows that \(\E{\dlant\p{Y}} = \E{\dl \p{Y}}\). To use \Cref{lem:ant_general} we require the following condition:
\begin{assumption}
	\label{assumpt:xmoms_exact}
	For some \(q > 2\),
	\[
		\E*{\abs*{X}^q} + \E*{\abs*{\pi\p*{Y}}^q} < \infty.
	\]
	Moreover, there exists constants \(\bar\delta,\, \bar\rho > 0\) such that \(x < \bar\delta\) implies
	\[
		\prob*{\abs*{\nestE[1]\p{Y} - \tpayoff} < x} \le \bar\rho x.
	\]
\end{assumption}
This condition implies the following result, which generalises \cite[Proposition 2.5]{Bourgey20} to higher moments beyond the variance. 
\begin{proposition}
	\label{lem:antfexact}
	Consider the antithetic estimator \(\dlant\p{Y}\) defined by \eqref{eqn:antf} and let \Cref{assumpt:xmoms_exact} hold for some \(q > 2\). There is \(b_1 > 0\), independent of \(\ell\), such that for \(1\le p\le q\) we have
	\[
	\E*{\abs*{\dlant\p{Y}}^p} \le b_1 2^{-q\p{1+p^{-1}}p\ell/2\p{q+1}}. 
	\]
\end{proposition}
\begin{proof}
	By mapping \(X\mapsto X - \pi\p{Y}\) if necessary, it can be assumed without loss of generality that \(\pi\p{Y} \equiv 0\). By the discrete Burkholder-Davis-Gundy inequality and under \Cref{assumpt:xmoms_exact}, it follows that there is \(c_q > 0\) depending only on \(q\) such that {\cite[Lemma 1]{gg19}, \cite[Lemma 2.5]{GilesHajiAli:2018}}
	\[
		\E*{\abs*{\appE[\ell - 1]\p{Y} - \nestE[1]\p*{Y}}^q} \le c_q \E*{\abs{X}^q} 2^{-q \ell / 2}.
	\]
	By \Cref{assumpt:xmoms_exact}, the conditions of \Cref{lem:ant_general} are met, which then provides the result.
\end{proof}

From \Cref{lem:antfexact}, it follows immediately that \(\var{\dlant\p{Y}}\) is of order \(2^{-3q\ell / 2\p{q+1}}\). Therefore, \(\var{\dlant\p{Y}}\text{Cost}\p{\dlant\p{Y}}\) is of order \(2^{\p{2 - q}\ell/2\p{q+1}}\). When \(q>2\), from MLMC complexity theory the estimator \eqref{eqn:mlmc} with correction terms \(\dlant\p{Y}\) has order \({\tol^{-2}}\) cost to attain root mean square error \(\tol\). In the case where \(q = 2\) and \(\E{\abs{X}^{2+\delta}} + \E{\abs{\pi\p{Y}}^{2+\delta}} = \infty\) for all \(\delta > 0\), it follows that \(\var{\dlant\p{Y}}\text{Cost}\p{\dlant\p{Y}}\) is of order \(1\) and the cost of estimating \(\nestE[0]\) is again of order \(\tol^{-2}\p{\log\tol}^2\). The antithetic differences \(\dlant\p{Y}\) rely on higher order moments to obtain lower variance. Hence, the condition \(q >2\) in \Cref{assumpt:xmoms_exact} is necessary to improve the asymptotic MLMC cost.  \\

 	\section{Antithetic Multilevel Monte Carlo}
		\label{sec:ant_approx}
		We now consider cases where either \(X\), \(Y\) or both require approximate sampling. In this setup, \(\appE\) is replaced with an MLMC estimator utilising multilevel simulation of \(X\). In addition, the correction terms \(\dlant\p{Y}\) are modified to account for varying levels of approximation in \(Y\), while retaining the telescoping property in \eqref{eqn:mlmc}.
		\subsection{Approximate Simulation of \(X\)}
			\label{sec:twostageexact}
			Consider for now the case where we have access to an exact simulator of \(Y\), but rely on a hierarchy of increasingly accurate approximations \(X\p{Y}\approx \Xl\p{Y}\) for integers \(k\ge 0\). Let the correction terms \(\dl[k]X = \dl[k]X\p{Y}\) be constructed such that
\[
	\E*{\dl[k]X\given Y} = 
	\begin{cases} 
	\E*{\Xl - \Xl[k - 1]\given Y} & k > 0\\
	\E*{\Xl[0]\given {Y}} & k = 0
	\end{cases}.
\]
To avoid confusion between the outer and inner layers of approximation, \(\ell\) is used to refer to the approximation of \(\nestE[1]\p{Y}\), and \(k\) denotes the approximation level of \(X\) given \(Y\). The following condition controls the error and cost of estimating \(\nestE[1]\p{Y}\): 
\begin{assumption}
	\label{assumpt:mlmcrates}
	There are constants \(a_0, a_1, a_2> 0\) and \(q > 2\) with \(\beta_p \ge 1\) and \(\lambda_p \ge 0\) defined for all \(p\le q\), each independent of \(k\), such that {for all \(k\ge0\)}
	\[
		\begin{aligned}
			\textnormal{Cost}\p*{X_k\p{Y}} &\le a_0 2^{k}\\
			\E*{\abs*{X_k\p{Y} - X\p{Y}}^q} &\le a_1 k^{\lambda_q} 2^{-qk/2}\\
			\E*{\abs*{\dl[k] X\p{Y}}^p} &\le a_2 k^{\lambda_p} 2^{-\beta_ppk/2}.
		\end{aligned}
	\]
\end{assumption}
{In certain cases, it is enough to take \(\dl[k]X = X_k - X_{k-1}\) to satisfy the above assumption. However, it is occasionally beneficial to use an alternative, unbiased estimator of \(X_k - X_{k-1}\) to increase \(\bq[p]\). An example where this is true is the use of the antithetic correction term \eqref{eqn:antf}, which has smaller moments than \eqref{eqn:dl_mcdiff}.} {The logarithmic factor \(\lambda_p\) is typically zero for examples considering only two nested expectations. However, non-trivial values of \(\lambda_p\) can arise when \(X\) contains further nested expectations, see \Cref{rem:two_stage_prop_recursive}. We keep the conditions in \Cref{assumpt:mlmcrates} fairly general to include many settings found in the literature including SDEs \cite{Giles2008MLMC}, PDEs \cite{Cliffe:2011,Elfverson:2016selectiverefinement}, further nested Monte Carlo estimates \cite{Bourgey20,GilesHajiAli:2018,syed2023optimal,zhou2022}, particle systems \cite{bhst22,ht18,brps23} and biased estimators of such settings.}\\

At level \(\ell\), a multilevel estimate of \(\nestE[1]\p{y}\), given \(Y = y\), is constructed through 
\begin{equation}
	\label{eqn:appE}
	\begin{aligned}
	\appEml\p*{y} &\defeq \sum_{k = 0}^{\ell} \p*{\Mlk}^{-1}\sum_{n=1}^{\Mlk} \dl[k]X^{(n)}\p{y},\\
		\Mlk &= \max\br{\Mlk[0][0]2^{\ell - \zeta k},1}
	\end{aligned}	
\end{equation}
for fixed \(1 \le \zeta \le \bq\) and where, {for each \(k = 0,\dots, \ell\), \(\br{\dl[k]X^{(n)}\p{y}}_{n= 1}^{\Mlk}\) are conditionally independent samples of \(\dl[k]X\p{y}\) given \(Y = y\).} The cost of generating a single sample of \(\appEml\) is 
\begin{equation}
\label{eqn:cost_dlant_ml}
\textnormal{Cost}\p*{\appEml} = \sum_{k=0}^{\ell} \Mlk\textnormal{Cost}\p*{X_{k}} \le \p{\Mlk[0][0]+1}a_0 2^{\ell}\cdot
\begin{cases}
1 & \textnormal{if } \zeta > 1\\
\ell + 1 & \textnormal{if } \zeta = 1
\end{cases}.
\end{equation}

Similar to \eqref{eqn:antf}, we construct an antithetic multilevel correction term by sampling
\begin{equation}
	\label{eqn:ant_mlmc_terms}
	\begin{aligned}
\antEfml\p*{Y} &\defeq \frac{1}{2}\p*{\antEcml[0]\p*{Y} + \antEcml[1]\p*{Y}}\\
		&\quad +  {\Mlk[\ell][\ell]}^{-1}\sum_{n=1}^{\Mlk[\ell][\ell]} \dl X^{(n)}\p*{Y} ,
	\end{aligned}
\end{equation}
where \(\antEcml[0]\p*{Y}\) and \(\antEcml[1]\p*{Y}\) are conditionally independent samples of \(\appEml[\ell-1]\p*{Y}\) given \(Y\). The antithetic difference is then defined for \(\ell > 0\) by
\begin{equation}
\label{eqn:dlantfexacty}
\dlantml\p{Y} \defeq \max\br*{\antEfml\p*{Y}, \tpayoff} - \frac{1}{2}\sum_{i=0}^1 \max\br*{\antEcml[i]\p*{Y}, \tpayoff}.
\end{equation}

It is possible to construct a multilevel estimator for \(\nestE[0]\) in \eqref{eqn:nestexp2} by using the correction terms \eqref{eqn:dlantfexacty} within the MLMC average \eqref{eqn:mlmc}. The bias induced by truncating the inner MLMC sums results in the fine component \(\antEfml\p{Y}\) containing samples of \(\dl X\) which do not appear in the corresponding coarse estimates \(\antEcml[i]\p*{Y}\). {This additional bias requires attention in the following analysis.}\\

The following result, adapted from the proof of \cite[Proposition 3.6]{Dereich2019}, is used frequently in what follows to analyse the antithetic MLMC difference:
\begin{lemma}
	\label{prop:zn_ind_ml}
	{Let \(\br{Z_k}_{k\in\nset}\) be a sequence of {\(\rset\)-valued} random variables with \(\E{Z_k} = 0\) and \(\E{Z_k^q} < \infty\) for some \(q \ge 2\). Let \(\br{\Ml[k]}_{k\in\nset}\) be a sequence of positive integers and let \(\{Z_k^{\p{n}}\}_{n=1}^{\Ml[k]}\) be independent samples of \(Z_k\) for each \(k\ge 0\).
	Then, there is \(c_q > 0\), depending only on \(q\), such that
	\[
	\E*{\abs*{\sum_{k=0}^\ell \Ml[k]^{-1} \sum_{n=1}^{\Ml[k]}  Z_k^{\p{n}}}^q} \le c_q\abs*{\sum_{k=0}^\ell \Ml[k]^{-1} \E*{\abs{Z_k}^q}^{2/q} }^{q/2}.
	\]}
\end{lemma}
\begin{proof}
	By the discrete Burkholder-Davis-Gundy inequality there is \(c_q > 0\) depending only on \(q\) such that
	\[
	\begin{aligned}
		\E*{\abs*{\sum_{k=0}^\ell \Ml[k]^{-1} \sum_{n=1}^{\Ml[k]}  Z_k^{\p{n}}}^q}^{2/q}
		&\le c_q^{2/q}\E*{\abs*{\sum_{k=0}^\ell \abs*{\Ml[k]^{-1} \sum_{n=1}^{\Ml[k]}  Z_k^{\p{n}}}^2}^{q/2}}^{2/q}
	\end{aligned}
	\]
	Then, applying the Minkowski and Burkholder-Davis-Gundy inequalities for \(q\ge 2\) gives
	\[
		\begin{aligned}
		\E*{\abs{\overline{Z_{N_0,\dots, N_\ell}}}^q}^{2/q}
		&\le c_q^{2/q}\sum_{k=0}^\ell \E*{\abs*{\Ml[k]^{-1} \sum_{n=1}^{\Ml[k]} Z_k^{\p{n}}}^{q}}^{2/q}\\
		&\le c_q^{4/q}\sum_{k=0}^\ell \Ml[k]^{-1} \E*{\abs{Z_k}^q}^{2/q}.
		\end{aligned}
	\]
\end{proof}

From \Cref{prop:zn_ind_ml}, we arrive at the following bound on the error of the inner MLMC estimator \eqref{eqn:appE}.
\begin{lemma}
	\label{lem:2_stage_moments}
	Let \Cref{assumpt:mlmcrates,assumpt:xmoms_exact} hold for some \(q > 2\) with \(\lambda_q\ge 0\) and \(\bq \ge 1\). Consider the MLMC estimator \(\appEml\p*{Y}\) given by \eqref{eqn:appE} with inner sample sizes scaled according to \(1\le \zeta\le \bq\). Then, there is \(b_2 > 0\), independent of \(\ell\), such that 
	\[
		\E*{\abs*{\appEml\p*{Y} - \nestE[1]\p*{Y}}^q} \le b_2 \ell^{\lambda_q} 2^{-q\ell / 2} 
		\begin{cases}
			1			&	\zeta < \bq\\
			\ell^{q/2}	&	\zeta = \bq 
		\end{cases}.
	\]
\end{lemma}
\begin{proof}
	We have 
	\[
	\begin{aligned}
	\E*{\abs*{\appEml\p*{Y}  - \nestE[1]\p*{Y}}^q} &\le 2^{q-1}\bigg(
	\E*{\abs*{\appEml\p*{Y}  - \E*{X_{\ell}\given Y}}^q}  \\
	&\qquad + 
	\E*{\abs*{\E*{X_{\ell}\given Y} - \nestE[1]\p*{Y}}^q}\bigg)
	\end{aligned}
	\]
	For the first term, writing \(\E{X_{\ell}\given Y} = \sum_{k=0}^{\ell}\E{\dl[k] X\given Y}\) and using \Cref{prop:zn_ind_ml}
	\[
	\begin{aligned}
	\E*{\abs*{\appEml\p*{Y} - \E*{X_{\ell}\given Y}}^q} &= \E*{\abs*{\sum_{k=0}^{\ell} \p*{\Mlk[\ell]}^{-1}\sum_{n=1}^{\Mlk[\ell]}\p*{\dl[k]X^{\p{n}}\p{Y} - \E{\dl[k]X\given Y}} }^q}\\ 
	&\le c_q^{2}2^{q-1}\abs*{\sum_{k=0}^{\ell-1}\p*{\Mlk}^{-1}\E*{\abs{\dl[k]X}^q}^{2/q}}^{q/2}\\
	&\le c_q^{2}2^{q-1} \ell^{\lambda_q} 2^{-q\ell/2}
	\begin{cases}
	1			&	\zeta < \bq\\
	\ell^{q/2}	&	\zeta = \bq 
	\end{cases},
	\end{aligned}
	\]
	using \Cref{assumpt:mlmcrates} in the final step. On the other hand, combining the Tower Property, Jensen's inequality and \Cref{assumpt:mlmcrates} gives
	\[
	\begin{aligned}
	\E*{\abs*{\E*{X_{\ell}\given Y} - \nestE[1]\p*{Y}}^q} &\le \E*{\abs*{X_{\ell} - X}^q}\\
	&\le a_1 \ell^{\lambda_q} 2^{-q\ell/2},
	\end{aligned}
	\]
	completing the result, with \(b_2 = 2^{q-1}\p*{2^{q-1}c_q^2 + a_1}\).
\end{proof}

As in \Cref{sec:twostageexactmc}, the antithetic difference can improve the convergence of the multilevel difference in the following sense.
\begin{proposition}
	\label{lem:2stageant}
	Let \Cref{assumpt:mlmcrates,assumpt:xmoms_exact} hold for some \(q > 2\) and \(\bq\ge 1\). Consider the antithetic multilevel correction terms \(\dlantml\p{Y}\) defined by \eqref{eqn:ant_mlmc_terms} and \eqref{eqn:dlantfexacty} with inner sample sizes \(\Mlk \propto 2^{\ell-\zeta k}\) for some \(1\le \zeta\le \bq\). Then, there is \(b_3 > 0\), independent of \(\ell\), such that for \(0<p\le q\) we have
	\begin{equation}
	\label{eqn:dlantmlbound}
	\E*{\abs*{\dlantml\p{Y}}^p}  \le b_3 \max\br*{\ell^{\lambda_p}2^{-\bq[p]\ell p / 2}, \ell^{\mu_p} 2^{-q\p{1+p^{-1}}p\ell/2\p{q+1}}},
	\end{equation}
	where
	\begin{equation}
		\label{eqn:mu_p}
		\mu_p \defeq \frac{{p+1}}{{q+1}}\lambda_q + 
		\begin{cases}
			0 & \zeta < \bq\\
			q\p{p+1}/2\p{q+1}	&	\zeta = \bq 
		\end{cases}.
	\end{equation}
\end{proposition}
This result follows immediately from the {more general} proof of \Cref{lem:2stageant_fixedy} in the following section. Note that the convergence {rate depends on} \(\bq[p]\) from \Cref{assumpt:mlmcrates}. If the differences \(\dl[k]X\) converge with \(\bq[p] = 1\), then \(\E*{\abs*{\dlantml\p{Y}}^p}\) is again of order \(2^{-\ell p / 2}\) {up to logarithmic factors} regardless of the antithetic difference. 
\begin{remark}
	\label{rem:two_stage_prop_recursive}
	A particular consequence of \Cref{lem:2stageant} is that the correction terms \(\dlantml\p{Y}\) and inner MLMC estimates \(\appEml\p{Y}\) themselves satisfy the conditions of \Cref{assumpt:mlmcrates} {with 
	\begin{equation}
	\label{eqn:bq_recursive}
	\begin{aligned}
		\lambda_p &\mapsto 
		\begin{cases}
			\lambda_p & \bq[p] < \frac{q}{q+1}\frac{p+1}{p}\\
			\mu_p 	& \textnormal{otherwise}
		\end{cases},\\
		\bq[p] &\mapsto \min\br*{\bq[p],\frac{q}{q+1}\frac{p+1}{p}}.
	\end{aligned}
	\end{equation}}
	This fact can be applied recursively to consider the case where \(X = \max\br{\nestE[2]\p{Y^\prime}, \pi\p{Y^\prime}}\) for  \(\nestE[2]\p{Y^\prime} = \E{X^\prime\given Y^\prime}\). Following \cite{zhou2022,syed2023optimal}, {it is possible to extend the above approach to deal with arbitrarily many nested expectations. We leave a rigorous analysis of this approach to future work.}
\end{remark}

\Cref{lem:2stageant} can be used to bound the cost of estimating \eqref{eqn:nestexp2} using {nested} MLMC simulation with exact simulation of \(Y\). A rigorous bound on the cost follows from the more general result in \Cref{cor:two_stage_complexity}.\\

{
	An alternate approach is employed in \cite{GilesHajiAliSpence:CVA,GilesHajiAli:2019sampling,zhou2022,syed2023optimal}, which each create an unbiased estimator of \(\nestE[1]\) by randomising the correction terms \(\dl[k]X\p{Y}\). The approach uses the observation \(\nestE[1]\p{Y} = \E*{\dl[k^\prime]X\p{Y}\ p_{k^\prime}^{-1}\given Y}\), where \(k^\prime\) is a random integer satisfying \(\prob{k^\prime = k} \propto 2^{-\zeta k}\) \cite{Rhee2015_Unbiased_estimation}. The methods in \Cref{sec:twostageexactmc} can then be applied to the unbiased estimator to approximate \(\nestE[0]\). In particular, from a batch of \(\Ml[\ell]\propto 2^{\ell}\) independent samples of \(\dl[k^\prime]X\p{Y}\ p_{k^\prime}^{-1}\), the proportion of samples which include realisations of \(\dl[k]X\) is of order \(2^\ell\E{\I{k^{\prime} = k}} \propto 2^{\ell-\zeta k}\). Hence, the parameter \(\zeta\) for the randomized probabilities has the same role as in \eqref{eqn:appE}. Under the conditions of \Cref{assumpt:mlmcrates} with \(q>2\) and for some \(2\le r \le q\), to ensure \(\E{\textnormal{Cost}\p{\dl[k^\prime]X\p{Y}}} < \infty\) and \(\E{\abs{\dl[k^\prime]X\p{Y} p_{k^\prime}^{-1}}^r} < \infty\) we require \cite[Remark 3.1]{GilesHajiAli:2019sampling}
	\[
		1 < \zeta < \frac{r}{r-1}\frac{\bq[r]}{2}.
	\]
For practical applications in which \(\bq[r]\le 2\), this inequality can restrict the maximum number of moments of the nested estimator. Particularly, ensuring the unbiased estimates have finite cost and variance necessitates \(\bq[2] > 1\). In contrast, under the conditions of \Cref{lem:2stageant}, the nested (biased) MLMC estimator has finite cost and \(q^\text{th}\) moment provided \(1\le \zeta\le \bq\). Furthermore, the biased estimate is affected by taking \(\zeta = \bq\) only up to an additional logarithmic factor in the rate \eqref{eqn:mu_p}, which has an asymptotic impact on the bound \eqref{eqn:dlantmlbound}, for \(p\le q\),  only when \(\bq[p] \ge q\p{1 + p^{-1}}/\p{q+1}\). Meanwhile, setting \(\zeta = 1\) within the biased estimate adds a logarithmic factor to the cost \eqref{eqn:cost_dlant_ml}.
	Hence, by avoiding the randomisation of \(k\), more control is gained over higher moments of the antithetic correction terms. Additionally, the biased method has finite cost and variance in the case when \(\bq[p] = 1\) for all \(2\le p\le q\) as is typical, for example, with Euler-Maruyama discretisation of an SDE.
}

		\subsection{Approximate Simulation of \(Y\)}
			\label{sec:twostageapprox}
			In many problems of interest, both \(X\) and \(Y\) depend on an underlying process which requires approximation. For an example, let \(\br{S_t}_{0\le t\le T}\) be the solution to an SDE which requires approximate sampling. Many problems of interest can be expressed in the form \eqref{eqn:nestexp2} for \(Y = S_{\tau}\) and \(X = \pi\p{S_T}\) depend on the underlying SDE at intermediate time \(\tau\) and maturity \(T\), respectively \cite{GilesHajiAliSpence:CVA,GilesHajiAli:2019sampling,zhou2022,syed2023optimal}. \\

The results of \Cref{sec:twostageexact} can be naturally extended to this more general case by altering \Cref{assumpt:mlmcrates} to account for approximate samples of \(Y\):
\begin{assumption}
	\label{assumpt:twostagegenY}
	There is positive constants \(a_0, a_1, a_2\)  and \(\lambda_p \ge 0, \bq[p] \ge 1\) defined for \(p\le q\), each independent of \(\ell\) and \(k\) such that for all \(k\ge 0\)
	\[
	\begin{aligned}
	\textnormal{Cost}\p*{X_k\p{{y}}} &\le a_0 2^k \quad {\textnormal{for all } y\in\rset^d}\\
	\E*{\abs*{X_k\p{Y_\ell} - X\p{Y_\ell}}^q} &\le a_1k^{\lambda_q} 2^{-qk/2}\\
	\E*{\abs*{\dl[k]X\p{Y_\ell}}^p} &\le a_2k^{\lambda_p} 2^{-\bq[p]pk/2}.
	\end{aligned}
	\]
\end{assumption}

{Since \(Y\) requires approximation, we must impose additional bounds on the convergence of \(Y_\ell\). Errors in approximating \(Y\) are carried into approximations of \(X\). Thus, we are interested in bounding how errors in \(Y\) affect errors in approximations \(X_k\). In particular, we consider the double differences 
\[
	\begin{aligned}
	\dl{\dl[k]{X\p{Y_\ell; Y_{\ell-1}}}} &\defeq \dl[k]{X\p{Y_\ell}}  - \dl[k]{X\p{Y_{\ell-1}}}.
	\end{aligned}
\]}
Suppose that the error made in approximating \(X\) does not affect the quality of approximations \(Y\) in the following sense:
\begin{assumption}
	\label{assumpt:twostageY}
	There is \(\hat a_0, a_3, a_4, a_5 > 0\) and \(\bqy[p]\ge 1\) {depending on \(p\le q\), each independent of \(\ell, k\)} such that for all \(\ell, {k} \ge 0\) 
	\begin{align}
			\textnormal{Cost}\p*{Y_\ell} &\le \hat a_0 2^\ell\nonumber\\
			\E*{\abs*{\tpayoff[Y] - \tpayoff[Y_{\ell}]}^p} &\le a_32^{-\bqy[p]p\ell/2}\label{eqn:y_pi_err}\\
			\E*{\abs*{X_k\p*{Y} - X_k\p*{Y_{\ell}}}^p} &\le a_42^{-\bqy[p]p\ell/2}\label{eqn:dxk_y}\\
			{\E*{\abs*{\dl{\dl[k]{X\p{Y_\ell; Y_{\ell-1}}}}}^p}} &\le {a_5 2^{-\bqy[p]p\ell/2}}.\label{eqn:double_diff_bound}
	\end{align}
\end{assumption} 
The inequality \eqref{eqn:y_pi_err} imposes an implicit condition on \(\pi\). Typically, a bound is available of the form
\[
	\E*{\norm*{Y - Y_{\ell}}^p} \le \hat a_32^{-\bqy[p]p\ell/2},
\]
where \(\norm{\cdot}\) represents the Euclidean norm here and throughout the paper.
Provided \(\pi\) is Lipschitz or sufficiently smooth with bounded higher order derivatives, then \eqref{eqn:y_pi_err} follows immediately from the bound on \(\norm{Y-Y_\ell}\). {The condition in \Cref{eqn:dxk_y} forms a stability assumption (in the \(L^p\) sense) on the approximations \(X_k\) given an error in the intermediate value \(Y\approx Y_\ell\). Additionally, \eqref{eqn:double_diff_bound} places an extra condition on the correction terms \(\dl[k]X\p{y}\), {which is required in the proof of \Cref{lem:2stageanty} to bound the error between \(\appE\p{Y}\) and \(\appE\p{Y_\ell}\)}. If \(\dl[k]X\p{y}\) {is a difference of} samples \(X_k\p{y}\) and \(X_{k-1}\p{y}\), then \eqref{eqn:dxk_y} implies \eqref{eqn:double_diff_bound} by {the triangle} inequality. Note that under \Cref{assumpt:mlmcrates,assumpt:twostagegenY,assumpt:twostageY} it follows from Jensen's inequality that there is \(a_6 > 0\) independent of \(\ell\) such that for any \(\ell, k \ge 0\)
\begin{equation}
	\label{eqn:dx_y}
	\begin{aligned}
	\E*{\abs{X\p{Y} - X\p{Y_\ell}}^q} &\le 3^{q-1}\bigg( \E*{\abs{X\p{Y} - X_k\p{Y}}^q}\\
	&\qquad + \E*{\abs{X_k\p{Y} - X_k\p{Y_\ell}}^q}\\
	&\qquad +   \E*{\abs{X_k\p{Y_\ell} - X\p{Y_\ell}}^q} \bigg)\\
	&\le a_6\p*{k^{\lambda_q}2^{-q k /2} + 2^{-\bqy q \ell / 2}}\\
	&\le 2 a_62^{-q \ell / 2}.
	\end{aligned}
\end{equation}
The last line follows from taking \(k\) large enough so that \(k^{\lambda_q}2^{-q k /2}\le 2^{-q\ell/2}\) and since \(\bqy\ge 1\) in \Cref{assumpt:twostagegenY}. This bound is useful in \Cref{lem:2stageant_fixedy}.\\}

For the SDE example, let \(X_{k}\p*{Y_\ell} = \pi\p{S_{T, \ell,k}}\) approximate \(S_T\) using Euler-Maruyama or Milstein discretisation with order \(2^{\ell}\) steps to approximate \(Y = S_\tau\) on \([0, \tau]\)  and \(2^{k}\) steps on \([\tau, T]\). Then \Cref{assumpt:twostageY,assumpt:twostagegenY} hold with \(\bq[p] = \bqy[p] = 1\) for Euler-Maruyama and \(\bq[p] = \bqy[p] = 2\) for Milstein discretisation under commutativity and growth conditions on the coefficients of the underlying SDE and a Lipschitz condition on \(\pi\p{\cdot}\). In further applications, such as when \(X\) depends on deeper nested expectations, it may happen that \(\bq[p]\) differs from \(\bqy[p]\). \\

{In some settings, such as when Milstein simulation of an underlying SDE requires sampling L\'evy areas, the value of \(\bqy[p]\) can be limited within \Cref{assumpt:twostageY} \cite{giles14antmilstein,cc80}. In the next section, we consider a different set of assumptions which use antithetic approximations of \(Y\) that can be used to obtain larger values of \(\bqy[p]\) in certain practical examples.\\}

Let \(Y_{\ell}\) and \(Y_{\ell-1}\) be a sufficiently correlated pair of approximations to \(Y\) satisfying \Cref{assumpt:twostageY}. For \(\ell^\prime\in \br{\ell-1,\ell}\) define
\begin{equation}
	\label{eqn:nested_mlmc_y}
	\begin{aligned}
		\antEcmly[j]\p{Y_{\ell^\prime}} &\defeq \sum_{k=0}^{\ell-1}\frac{1}{\Mlk[\ell][k][1]}\sum_{n=1}^{\Mlk[\ell][k][1]} \dl[k]{X^{(j, n)}\p{Y_{\ell^\prime}}} \\
		\antEfmly\p{Y_{\ell^\prime}} &= \frac{1}{2}\p*{\antEcmly[0]\p{Y_{\ell^\prime}}+\antEcmly[1]\p{Y_{\ell^\prime}}}\\
		&\quad + \frac{1}{\Mlk[\ell][\ell][1]}\sum_{n=1}^{\Mlk[\ell][\ell][1]} \dl{X^{(n)}\p{Y_{\ell^\prime}}},
	\end{aligned}
\end{equation}
{where \(\br{\dl[k]{X^{(j, n)}\p{Y_{\ell^\prime}}}}_{k, j, n}\) are conditionally independent samples of \(\dl[k]{X\p{Y_{\ell^\prime}}}\) and \(\br{\dl{X^{(n)}\p{Y_{\ell^\prime}}}}_n\) are conditionally independent samples of \(\dl{X\p{Y_{\ell^\prime}}}\), each satisfying \eqref{eqn:double_diff_bound}.}
In this way, consider the antithetic correction term
\begin{equation}
\label{eqn:dlantfy}
\begin{aligned}
\dlantmly\p*{Y_\ell, Y_{\ell-1}} &\defeq \max\br*{\antEfmly\p{Y_{\ell}}, \tpayoff[Y_\ell]}\\
&\quad - \frac{1}{2}\sum_{i=0}^1\max\br*{\antEcmly[i]\p*{Y_{\ell-1}}, \tpayoff[Y_{\ell-1}]},
\end{aligned}
\end{equation}
which resembles \eqref{eqn:dlantfexacty} with a slight offset caused by approximating \(Y\) at level \(\ell\) for the fine term and at level \(\ell-1\) for both coarse terms. Since \(Y_\ell\) is sampled only once for each inner MLMC approximation, it follows that
\begin{equation}
	\label{eqn:dlmly_cost}
	\textnormal{Cost}\p*{\dlantmly} \le a_02^\ell +  \Mlk[0][0]a_0 2^{\ell}\cdot
	\begin{cases}
	1&\zeta > 1\\
	\ell + 1 & \zeta = 1
	\end{cases}.
\end{equation}
Before extending \Cref{lem:2stageant} to the correction terms \(\dlantmly\), we first prove the following result, which is a generalisation of \Cref{lem:2stageant} using the same approximation of \(Y\) for the fine and coarse estimators. {Note that the following result does not impose \Cref{assumpt:twostageY} and instead relies on the conclusion of \eqref{eqn:dx_y}. This is helpful in the following section, which considers an alternative set of assumptions.}
\begin{proposition}
	\label{lem:2stageant_fixedy}
	Let \Cref{assumpt:xmoms_exact,assumpt:mlmcrates,assumpt:twostagegenY} hold for some \(q > 2\) and \(\bq\ge 1\). {Assume further that \(\E{\abs{X\p{Y} - X\p{Y_\ell}}^q}\le a_6 2^{-q \ell / 2}\) }for some \(a_6 > 0\). Consider the antithetic multilevel correction terms \(\dlantml\) defined by \eqref{eqn:ant_mlmc_terms} and \eqref{eqn:dlantfexacty} with inner sample sizes \(\Mlk \propto 2^{\ell-\zeta k}\) for some \(1\le \zeta\le \bq\). Then, there is \(b_3 > 0\), independent of \(\ell\), such that for \(0<p\le q\) we have
	\[
	\E*{\abs*{\dlantml\p{Y_\ell}}^p}  \le b_3 \max\br*{\ell^{\lambda_p}2^{-\bq[p]\ell p / 2}, \ell^{\mu_p} 2^{-q\p{1+p^{-1}}p\ell/2\p{q+1}}},
	\]
	where
	\[
	\mu_p \defeq \frac{{p+1}}{{q+1}}\lambda_q + 
	\begin{cases}
	0 & \zeta < \bq\\
	q\p{p+1}/2\p{q+1}	&	\zeta = \bq 
	\end{cases}.
	\]
\end{proposition}
\begin{proof}
	As in the proof of \Cref{lem:antfexact}, assume without loss of generality (taking \(X \mapsto X - \pi\p{Y}\)) that \(\pi\p{Y} = 0\). Define 
	\[
	\begin{aligned}
	\antEfmlun\p{Y_\ell} &\defeq \frac{1}{2}\p*{\antEcml[0]\p{Y_\ell} + \antEcml[1]\p{Y_\ell}}\\
	\overline{\Delta}_\ell^{\p{\textnormal{ML}}}\p{Y_\ell} &\defeq \max\br*{\antEfmlun\p{Y_\ell}, 0} - \frac{1}{2}\sum_{i=0}^1\max\br*{\antEcml[i]\p{Y_\ell}, 0}
	\end{aligned}
	\]
	so that \(\overline{\Delta}_\ell^{\p{\textnormal{ML}}}\p{Y_\ell}\) vanishes whenever \(\antEcml[0]\p{Y_\ell}\) and \(\antEcml[1]\p{Y_\ell}\) fall on the same side of \(0\) as in \Cref{sec:twostageexactmc}. Then,
	\[
	\begin{aligned}
	\E*{\abs*{\dlantml\p{Y_\ell}}^p} &\le 2^{p-1}\bigg(\E*{\abs*{\max\br*{\antEfml\p{Y_\ell}, 0} - \max\br*{\antEfmlun\p{Y_\ell}, 0}}^p } \\
	&\quad + \E*{\abs*{\overline{\Delta}_\ell^{\p{\textnormal{ML}}}\p{Y_\ell}}^p}\bigg).
	\end{aligned}
	\]
	The former term satisfies
	\[
	\begin{aligned}
	\E*{\abs*{\max\br*{\antEfml, 0} - \max\p*{\antEfmlun, 0}}^p }  &\le \E*{\abs*{{\antEfml} - {\antEfmlun}}^p }\\
	&\le \E*{\abs*{\dl X\p{Y_\ell}}^p}\\
	&\le a_2 \ell^{\lambda_p} 2^{-\beta_pp\ell/2},
	\end{aligned}
	\]
	where we use \Cref{assumpt:twostagegenY} in the final step. On the other hand, by \Cref{assumpt:xmoms_exact,assumpt:twostagegenY} and \Cref{lem:2_stage_moments}, the antithetic difference \(\overline{\Delta}_\ell^{\p{\textnormal{ML}}}\p{Y}\) satisfies the conditions of \Cref{lem:ant_general}. Hence, it follows for \(\mu_p\) as above that
	\[
	\begin{aligned}
	&\E*{\abs*{\overline{\Delta}_\ell^{\p{\textnormal{ML}}}\p{Y_\ell}}^p}\\
	 &\le b_0 \E*{\abs*{\appE[\ell][1]\p*{Y_\ell} - \nestE[1]\p*{Y}}^q}^{\p{p+1}/\p{q+1}}\\
	&\le b_02^{q-1}\p*{\E*{\abs*{\appE[\ell][1]\p*{Y_\ell} - \nestE[1]\p*{Y_\ell}}^q} + \E*{\abs*{\nestE[1]\p*{Y_\ell} - \nestE[1]\p*{Y}}^q}}^{\p{p+1}/\p{q+1}}\\
	&\le b_02^{q-1} \p{b_2 + a_6}^{\p{p+1}/\p{q+1}} \ell^{\mu_p} 2^{-q\p{1 + p^{-1}}p\ell / 2\p{q+1}}.
	\end{aligned}
	\]
	{In the final line above we use \Cref{lem:2_stage_moments} under \Cref{assumpt:twostagegenY} combined with the fact
	\[
		\begin{aligned}
			\E*{\abs*{\nestE[1]\p*{Y_\ell} - \nestE[1]\p*{Y}}^q} &= \E*{\abs*{\E*{X\p{Y_\ell} - X\p{Y}\given Y_,\ Y_\ell}}^q}\\
			&\le \E*{\abs*{X\p{Y_\ell} - X\p{Y}}^q}\\
			&\le a_62^{-q\ell/2}
		\end{aligned}
	\]
	by hypothesis, where the second line follows from Jensen's inequality and the tower property.}
\end{proof}
By setting \(Y = Y_\ell\), \Cref{lem:2stageant_fixedy} proves \Cref{lem:2stageant}.
The following result, based on \cite[Theorem A.2]{GilesHajiAliSpence:CVA}, extends \Cref{lem:2stageant_fixedy} to account for different approximations to \(Y\) at the fine and coarse levels.
\begin{proposition}
	\label{lem:2stageanty}
	Let \Cref{assumpt:xmoms_exact,assumpt:mlmcrates,assumpt:twostagegenY,assumpt:twostageY} hold for some \(q>2\), \(\bq[p], \bqy[p] \ge 1\) and \(\lambda_p \ge 0\). Consider the antithetic multilevel correction terms \(\dlantmly\) defined through \eqref{eqn:nested_mlmc_y} and \eqref{eqn:dlantfy} with inner sample sizes \(\Mlk\propto 2^{\ell-\zeta k}\) for some \(1\le \zeta\le \bq\). Then, there is \(b_4>0\), independent of \(\ell\), such that
	\[
	\E*{\abs*{\dlantmly}^p} \le b_4\max\br*{\ell^{\nu_p}2^{-\bqy[p]\ell p / 2}, \ell^{\lambda_p}2^{-\bq[p]\ell p / 2}, \ell^{\mu_p}2^{-q\p{1+p^{-1}}p\ell/2\p{q+1}}},
	\]
	for \(p\le q\), where \(\mu_p\) is given by \eqref{eqn:mu_p} and
	\begin{equation}
		\label{eqn:nu_p}
		\nu_p \defeq 
		\begin{cases}
			0	& \zeta = 1\\
			\ell^{p/2} &	\zeta > 1
		\end{cases}.
	\end{equation}
\end{proposition}
\begin{proof}
	The correction term \eqref{eqn:dlantfy} can be divided according to
	\[
	\begin{aligned}
		\dlantmly &= \dl{^{(1)}}  + \dlantml\p{Y_{\ell-1}}, \quad \textnormal{where}\\
		\dl{^{(1)}} &\defeq  \max\br*{\antEfmly\p{Y_{\ell}}, \tpayoff[Y_{\ell}]}  - \max\br*{\antEfmly\p{Y_{\ell-1}}, \tpayoff[Y_{\ell-1}]},\\
	\end{aligned}  
	\]
	and \(\dlantml\) is as in \eqref{eqn:dlantfy}.
	Then, by Jensen's inequality
	\[
	\E*{\abs*{\dlantmly}^p} \le 2^{p-1}\p*{\E*{\abs*{\dl{^{(1)}}}^p} + \E*{\abs*{\dlantml\p{Y_{\ell-1}}}^p} }.
	\]
	The term \(\E*{\abs*{\dlantml\p{Y_{\ell-1}}}^p}\) uses the same approximation to \(Y\) at the fine and coarse terms and can therefore be bounded using \Cref{lem:2stageant_fixedy} with \(Y_\ell\mapsto Y_{\ell-1}\) {and since \eqref{eqn:dx_y} holds under \Cref{assumpt:mlmcrates,assumpt:twostagegenY,assumpt:twostageY}}. Meanwhile, using the fact that
	\[
		\begin{aligned}
			&\abs*{\max\br*{x_0, y_0} - \max\br*{x_1, y_1}}\\
		    &\le \abs*{\max\br*{x_0, y_0} - \max\br*{x_1, y_0}} + \abs*{\max\br*{x_1, y_0} - \max\br*{x_1, y_1}}\\
			&\le \abs*{x_0 - x_1} + \abs*{y_0 - y_1},
		\end{aligned}
	\]
	 and applying Jensen's inequality, we have\
	\[
	\begin{aligned}
		&\E*{\abs*{\dl{^{(1)}}}^p}\\
		&\le 2^{p-1}\p*{\E*{\abs*{\antEfmly\p{Y_{\ell}} - \antEfmly\p{Y_{\ell-1}}}^p} + \E*{\abs*{\tpayoff[Y_\ell] - \tpayoff[Y_{\ell-1}]}^p}} \\
		&\le 2^{p-1}\E*{\abs*{\antEfmly\p{Y_{\ell}} - \antEfmly\p{Y_{\ell-1}}}^p} + 4^{p-1}a_3\ell^{\lambda_p}2^{-\bqy[p]\ell p/2}\\
	\end{aligned}
	\]
	using \Cref{assumpt:twostageY} to bound the latter expectation. Moreover, using \Cref{prop:zn_ind_ml}, it follows that
	\[
		\begin{aligned}
			&\E*{\abs*{\antEfmly\p{Y_{\ell}} - \antEfmly\p{Y_{\ell-1}}}^p}\\
		    &= \E*{\abs*{\sum_{k=0}^\ell \Mlk^{-1}\sum_{n=1}^{\Mlk}\dl{\dl[k]{X^{\p{n}}\p{Y_\ell; Y_{\ell-1}}}} }^p}\\
			&\le 2^{p-1}\E*{\abs*{\sum_{k=0}^\ell \Mlk^{-1}\sum_{n=1}^{\Mlk}\dl{\dl[k]{X^{\p{n}}\p{Y_\ell; Y_{\ell-1}}}} - \E*{\dl{\dl[k]{X\p{Y_\ell; Y_{\ell-1}}}}}}^p} \\
			&\qquad + 2^{p-1}\abs*{\sum_{k=0}^\ell \E*{\dl{\dl[k]{X\p{Y_\ell; Y_{\ell-1}}}}} }^p \\
			&\le 2^{p-1}c_q\abs*{\sum_{k=0}^\ell \Mlk^{-1}\E*{\abs*{\dl{\dl[k]{X\p{Y_\ell; Y_{\ell-1}}}}}^p}^{2/p}}^{p/2} + 2^{p-1}\abs*{\E*{X_\ell\p*{Y_\ell} - X_\ell\p*{Y_{\ell-1}}}}^p\\
			&\le 2^{p-1}a_42^{-\bqy[p]\ell p / 2} + 2^{p-1}c_q2^{-\bqy[p]\ell p / 2}
			\begin{cases}
				\ell^{p/2}  & \zeta > 1\\
				1	&	\zeta = 1
			\end{cases},
		\end{aligned}
	\]
	where we use
	\[
	\sum_{k=0}^\ell \E*{\dl{\dl[k]{X\p{Y_\ell; Y_{\ell-1}}}}} = \E{ X_{\ell}\p{Y_\ell} - X_{\ell}\p{Y_{\ell-1}}},
	\]
	and the final line follows from \eqref{eqn:double_diff_bound} and \eqref{eqn:dxk_y} in \Cref{assumpt:twostagegenY} and taking \(\Mlk = \max\br{1, 2^{\ell - \zeta k}}\).
\end{proof}

The previous results are enough to provide the following classification of the cost of estimating \eqref{eqn:nestexp2} using inexact samples of \(X\) and \(Y\). Since the logarithmic term \(\lambda_p\) in \Cref{assumpt:twostagegenY} often disappears for two nested expectations, we assume \(\lambda_p = 0\) in this result.
\begin{theorem}
	\label{cor:two_stage_complexity}
	Let \(\dlantmly\) be as in \eqref{eqn:dlantfy} and \Cref{assumpt:twostageY,assumpt:twostagegenY,assumpt:xmoms_exact,assumpt:mlmcrates} hold for some \(q > 2\), \(\bq[2],\bqy[2] \ge 1\) and \(\lambda_2 = 0\), with inner samples \(\Mlk\propto 2^{\ell-\zeta k}\).
	Then, for \(\tol > 0\) there is \(C > 0\) independent of \(\tol\) and parameters \(L, \{M_\ell\}_{\ell=0}^L\) such that  the MLMC approximation \eqref{eqn:mlmc} of \eqref{eqn:nestexp2} has root mean square error \(\tol\) with cost bounded by
	\[
	\costmlmc \le C\tol^{-2}
	\begin{cases}
	1 & \bq[2],\bqy[2]>1\\
	\abs*{\log\tol}^{3}	& \textnormal{otherwise}
	\end{cases}.
	\]
\end{theorem}
\begin{proof}
	From standard complexity {analysis} of MLMC \cite{Giles2008MLMC,Cliffe:2011}, we have
	\[
	\costmlmc \le \tol^{-2}\p*{\sum_{\ell=0}^L \sqrt{\textnormal{Cost}\p*{\dlantmly}\var*{\dlantmly} }}^2,
	\]
	where we require \(L\propto \log\p{\tol^{-1}}\) to ensure the squared bias of the MLMC estimate is smaller than \(\tol^2/2\). When \(\bq[2]>1\) and \(q>2\) we know from \eqref{eqn:cost_dlant_ml} and \Cref{lem:2stageanty} with \(1 \le \zeta \le \bq\) that the sum is convergent as \(L\to\infty\).\\
	
	{If \(\bqy[2] = 1\) or \(\bq[2] = 1\),  from \eqref{eqn:dlmly_cost} and \Cref{lem:2stageanty} we have 
	\[
	\textnormal{Cost}\p{\dlantmly\p{Y}}\var{\dlantmly\p{Y}} = \Order{\ell},
	\]
	 adding the extra factor \(\abs{\log\tol}^{3}\) to the cost.}
\end{proof}
 		
	\section{Doubly Antithetic Estimate}	
		\label{sec:ant_path_2stage}
		Many practical systems are well represented by SDEs {which have non-commutative diffusion coefficients in the sense of \cite{Kloeden:1999}}. This feature precludes efficient Milstein simulation and effectively limits convergence rates to those of Euler-Maruyama discretisation \cite{cc80}. If both \(X\) and \(Y\) in \eqref{eqn:nestexp2} depend on an SDE with this property, resorting to Euler-Maruyama discretisation implies \(\bqy[p] = \bq[p] = 1\) in \Cref{assumpt:twostageY,assumpt:twostagegenY} under standard conditions on the coefficients of the SDE \cite{Kloeden:1999}. For these values of the parameters, \Cref{cor:two_stage_complexity} implies the MLMC estimator of \(\nestE[0]\) discussed in the previous section is subject to a significant logarithmic factor in the cost. However, in the context of multilevel Monte Carlo, improving upon Euler-Maruyama convergence rates is often possible while avoiding the simulation of L\'evy areas. The approach in \cite{giles14antmilstein} uses an antithetic coupling of two fine SDE paths with a coarse SDE path to cancel lower-order terms from the L\'evy areas within the multilevel differences. Similar approacher are applied in \cite{hs23} to stochastic partial differential equations and in \cite{brps23,ht18} for interacting particle systems. Such methods can directly be applied to create antithetic terms \(\dl[k]X\p{Y}\) which improve the rate \(\bq[p]\) in \Cref{assumpt:mlmcrates,assumpt:twostagegenY}. Of particular interest here are antithetic couplings of \(Y\) and how these interact with the antithetic estimators discussed above.\\

Motivated by the results in \cite{giles14antmilstein}, the key contribution of this section is to {consider an alternative condition to \Cref{assumpt:twostageY}, which permits an additional antithetic difference from the approximation of \(Y\).} {Similar to \Cref{assumpt:twostageY}, we define the antithetic double difference
\[
\antdd \defeq \frac{1}{2}\sum_{i=0}^1 \dl[k]{X}\p*{\Yantf{\ell}{i}} - \dl[k]{X}\p*{\Yantc{\ell -1}}.
\] }
In particular, assume that for two fine path approximations \(\Yantf{\ell}{0}\) and \(\Yantf{\ell}{1}\), with corresponding coarse approximation \(\Yantc{\ell-1}\) the following condition holds:
\begin{assumption}
	\label{assumpt:ant_Y_path}
	For some \(q > 2\) and all \(p\le q\), there is \(\bqy[p] \ge 1\) such that
	\[
	\begin{aligned}
	\textnormal{Cost}\p*{Y_\ell} &\le a_0 2^\ell,\\
	\E*{\norm*{\frac{1}{2}\sum_{i=0}^1 \Yantf{\ell}{i} - \Yantc{\ell-1}}^p} &\le a_3 2^{-\bqy[p]\ell p / 2},\\
	\E*{\abs*{\frac{1}{2}\sum_{i=0}^1 X_k\p*{\Yantf{\ell}{i}} - X_k\p*{\Yantc{\ell-1}}}^p} &\le a_4 2^{-\bqy[p] \ell p  / 2},\\
	\E*{\abs*{\antdd}^q} &\le a_52^{-\bqy[p]\ell p / 2},\\
	\E*{\norm{Y_\ell - Y}^q} + \E*{\abs*{X\p{Y} - X\p{Y_\ell}}^q} &\le a_6 2^{-q\ell / 2}.\\
	\end{aligned}
	\]
\end{assumption}
{As remarked following \Cref{assumpt:twostageY}, if \(\dl[k]X\p{y}\) is a {difference between} samples of \(X_k\p{y}\) and \(X_{k-1}\p{y}\), the condition on \(\antdd\), required in the proof of \Cref{prop:moments_antithetic_path_twostage}, is an immediate consequence of {the triangle} inequality and the third condition in \Cref{assumpt:ant_Y_path}.\\}

By carefully coupling the two fine paths, it is often possible to increase \(\bqy[p]\) over what is possible with a single fine approximation as required by \Cref{assumpt:twostageY}.
In \cite{giles14antmilstein}, an antithetic Milstein coupling of SDE paths is developed and shown to satisfy \Cref{assumpt:ant_Y_path} with \(\bqy[p] = 2\) for all \(p< \infty\), while avoiding simulation of L\'evy areas under general assumptions.\\

It is necessary to bound moments of the antithetic difference in the payoffs \(\pi\) given by
\[
\Delta^{\p{\textnormal{ant}}}_\ell\pi\defeq \frac{1}{2}\sum_{i=0}^1 \pi\p{\Yantf{\ell}{i}} - \pi\p{\Yantc{\ell-1}}.
\] 
{\Cref{lem:ant_general} provides a useful result for antithetic differences of \(\rset\)-valued variables with \(\pi:\rset\to\rset\) defined by \(\pi\p{\cdot} = \max\br{\cdot, 0}\). The result relies on the fact that \(\max\br{\p{x_0 + x_1}/2, 0} - \p{\max\br{x_0, 0} + \max\br{x_1, 0}}/2\) is zero whenever \(x_0\) and \(x_1\) lie on the same side of 0. For \(\pi:\rset^d\to\rset\) and \(y\in\rset^d\), this motivates considering functions of the form \(\pi\p{y} = \max\br{P\p{y}, 0}\), where \(P:\rset^d\to\rset\) projects \(y\) into \(\rset\). If \(P\) is an affine projection, then \(\pi\) retains the key property that \(\max\br{\p{P\p{y_0} + P\p{y_1}}/2, 0} - \p{\max\br{P\p{y_0}, 0} + \max\br{P\p{y_1}, 0}}/2 = 0\) whenever \(P\p{y_0}\) and \(P\p{y_1}\) are on the same side of 0.  In general, assume there exists \(R\ge 0\) affine projections \(\{P_r\}_{r=1}^R\) mapping \(\rset^d\to\rset\) and an everywhere twice differentiable function \(\pi_s:\rset^d\to\rset\), with bounded second derivatives, such that \(\pi\) is of the form
\begin{equation}
\label{eqn:pi_ant_mil}
\pi\p*{y} = \pi_s\p*{y}  + \sum_{r = 1}^R \max\br*{P_r\p{y}, 0}.
\end{equation}
An extension of \Cref{lem:ant_general} to antithetic differences involving functions of \(\rset^d\)-valued random variables taking the form \eqref{eqn:pi_ant_mil} is proven in \Cref{lem:ant_general_smooth}. Similar to the hypothesis of \Cref{lem:ant_general}, it is important to bound the probability that \(P_r\p{Y}\) is sampled close to zero, which motivates the following assumption.}
\begin{assumption}
	\label{assumpt:ant_sde_density}
	Let \(\pi\) be of the form \eqref{eqn:pi_ant_mil} {for some affine projections \(\br{{P_r}}_{r=1}^R\).} There is \(\hat\delta, \hat\rho > 0\) such that \(x < \hat\delta\) implies 
	\[
	\sup_{1\le r\le R} \prob{\abs{{P_r\p{Y}}} \le x} \le \hat\rho x.
	\]
\end{assumption}
{The following result bounds moments of \(\Delta^{\p{\textnormal{ant}}}_\ell\pi\) under the setup defined above. In particular, the result extends \cite[Theorem 5.2]{giles14antmilstein} to include bounds on arbitrary moments.}
\begin{lemma}
	\label{lem:ant_mil_hinge}
	Let  \(\pi\) be defined through \eqref{eqn:pi_ant_mil} and consider the antithetic approximations given by \(\Yantf{\ell}{0}, \Yantf{\ell}{1}\) and \(\Yantc{\ell - 1}\) satisfying \Cref{assumpt:ant_Y_path} for some \(q > 2\) and \(\bqy[p]\) defined for \(p\le q\). Under \Cref{assumpt:ant_sde_density}, it follows that
	\[
	\E*{\abs*{\frac{1}{2}\p*{\pi\p{\Yantf{\ell}{0}} + \pi\p{\Yantf{\ell}{1}}} - \pi\p{\Yantc{\ell-1}} }^p} \le b_5 2^{-\min\br{\bqy[p], q\p{1 + p^{-1}}/\p{q+1}} p \ell / 2},
	\] 
	for all \(2\le p \le q\), where \(b_5\) is independent of \(\ell\).
\end{lemma}
\begin{proof}
	Let \(\overline{Y_\ell} \defeq \p{\Yantf{\ell}{0} + \Yantf{\ell}{1}}/2\) so that by Jensen's inequality 
	\[
	\begin{aligned}
	&\E*{\abs*{\frac{1}{2}\p*{\pi\p{\Yantf{\ell}{0}} + \pi\p{\Yantf{\ell}{1}}} - \pi\p{\Yantc{\ell-1}} }^p}\\
	&\le 
	2^{p-1}\p*{
		\E*{\abs*{\frac{1}{2}\p*{\pi\p{\Yantf{\ell}{0}} + \pi\p{\Yantf{\ell}{1}}} - \pi\p{\overline{Y_\ell}} }^p} + \E*{\abs*{\pi\p*{\overline{Y_\ell}} - \pi\p{\Yantc{\ell-1}}}^p} }.
	\end{aligned}	
	\]
	By \Cref{assumpt:ant_Y_path}, since \(\pi\) is Lipschitz it follows that 
	\[
	\E*{\abs*{\pi\p*{\overline{Y_\ell}} - \pi\p{\Yantc{\ell-1}}}^p} \le \hat b_5 2^{-\bqy[p]\ell p / 2}.
	\]
	On the other hand, from \Cref{lem:ant_general_smooth} under \Cref{assumpt:ant_Y_path,assumpt:ant_sde_density} it is shown that 
	\[
	\E*{\abs*{\frac{1}{2}\p*{\pi\p{\Yantf{\ell}{0}} + \pi\p{\Yantf{\ell}{1}}} - \pi\p{\overline{Y_\ell}} }^p} \le \tilde b_5 \E*{\norm{Y_\ell - Y}^q}^{\p{p+1}/\p{q+1}},
	\]
	which completes the result from the third inequality in \Cref{assumpt:ant_Y_path}.
\end{proof}

For the two stage problem \eqref{eqn:nestexp2}, the correction term \eqref{eqn:dlantfy} can be extended to account for antithetic path differences:
\[
\begin{aligned}
\dlantpath &\defeq \frac{1}{2}\sum_{i=0}^1\bigg(\max\br*{\antEfmly\p{\Yantf{\ell}{i}}, \tpayoff[\Yantf{\ell}{i}]}- \max\br*{\antEcmly[i]\p*{\Yantc{\ell-1}}, \tpayoff[\Yantc{\ell-1}]}\bigg).
\end{aligned}
\]
{The difference \(\dlantpath\) consists of two fine and two coarse components. The terms at approximation level \(\ell\) contain an antithetic coupling of fine approximations to \(Y\) satisfying \Cref{assumpt:ant_Y_path}. Meanwhile, the coarse components contain an antithetic coupling of inner MLMC approximations to \(\nestE[1]\) given a single coarse estimate of \(Y\), as studied in \Cref{sec:ant_approx}. The result of \Cref{lem:2stageanty} then follows for the doubly antithetic difference \(\dlantpath\):}
\begin{proposition}
	\label{prop:moments_antithetic_path_twostage}
	Under \Cref{assumpt:ant_Y_path,assumpt:twostagegenY,assumpt:ant_sde_density,assumpt:xmoms_exact}, defined for some \(q > 2\) and \(\bq[p], \bqy[p] \ge 1\), for \(\pi\) defined through \eqref{eqn:pi_ant_mil} it follows that there is \(b_6 > 0\) independent of \(\ell\) such that
	\[
	\E*{\abs*{\dlantpath}^p} \le b_6\max\br*{\ell^{\nu_p}2^{-\bqy[p]\ell p}, \ell^{\mu_p}2^{-\min\br{\bq[p], q\p{1+p^{-1}}/\p{q+1}}p\ell/2}},
	\]
	for all \(p\le q\), where \(\mu_p\) and \(\nu_p\) are defined in \eqref{eqn:mu_p} and \eqref{eqn:nu_p}.
\end{proposition}
\begin{proof}
	Define 
	\[
	\begin{aligned}
	\dl{^{(1)}} &\defeq \frac{1}{2}\sum_{i=0}^1\max\br*{\antEfmly\p{\Yantf{\ell}{i}}, \tpayoff[\Yantf{\ell}{i}]} - \max\br*{\antEfmly\p*{\Yantc{\ell-1}}, \tpayoff[\Yantc{\ell-1}]},
	\end{aligned}
	\]
	and let \(\dlantml\) be as in \eqref{eqn:dlantfy}
	so that 
	\[
	\E*{\abs*{\dlantpath}^p} \le 2^{p-1}\p*{\E*{\abs*{\dl{^{(1)}}}^p} +  \E*{\abs*{\dlantml\p{\Yantc{\ell-1}}}^p}}.
	\]
	{The latter term contains only a single approximation of \(Y\) and can be bounded using \Cref{lem:2stageant_fixedy}.
	On the other hand, defining}
	\[
	\begin{aligned}
	\antEfmlun &\defeq \frac{1}{2}\p*{\antEfmly\p{\Yantf{\ell}{0}} + \antEfmly\p{\Yantf{\ell}{1}}}\\
	\overline{\pi} &\defeq \frac{1}{2}\p*{ \pi\p*{\Yantf{\ell}{0}} + \pi\p*{\Yantf{\ell}{1}} },
	\end{aligned}
	\]
	we have 
	\[
	\begin{aligned}
	&\E*{\abs*{\dl{^{(1)}}}^p}\\
	&\le 2^{p-1}\E*{\abs*{\frac{1}{2}\sum_{i=0}^1\max\br*{\antEfmly\p{\Yantf{\ell}{i}}, \tpayoff[\Yantf{\ell}{i}]} - \max\br*{\antEfmlun, \bar\pi}  }^p}\\
	&\quad + 2^{p-1}\E*{\abs*{\max\br*{\antEfmlun, \bar\pi} - \max\br*{\antEfml\p*{\Yantc{\ell-1}}, \pi\p*{\Yantc{\ell-1}}} }^p}.
	\end{aligned}
	\]
	For the former term,
	\[
	\begin{aligned}
	&\E*{\abs*{\frac{1}{2}\sum_{i=0}^1\max\br*{\antEfmly\p{\Yantf{\ell}{i}}, \tpayoff[\Yantf{\ell}{i}]} - \max\br*{\antEfmlun, \bar\pi}  }^p}\\
	&= \E*{\abs*{\frac{1}{2}\sum_{i=0}^1\max\br*{\antEfmly\p{\Yantf{\ell}{i}}- \tpayoff[\Yantf{\ell}{i}], 0} - \max\br*{\antEfmlun- \bar\pi, 0}  }^p}\\
	&\le \hat b_6 \ell^{\mu_p} 2^{q\p{1+p^{-1}}\ell p / 2\p{q+1}},
	\end{aligned}
	\]
	using \Cref{lem:ant_general} under \Cref{assumpt:ant_Y_path,assumpt:twostagegenY,assumpt:xmoms_exact} in the final line.
	For the latter term,
	\[
	\begin{aligned}
	&\E*{\abs*{\max\br*{\antEfmlun, \bar\pi} - \max\br*{\antEfml\p*{\Yantc{\ell-1}}, \pi\p*{\Yantc{\ell-1}}} }^p}\\
	&\le \E*{\abs*{\antEfmlun - \antEfml\p*{\Yantc{\ell-1}} }^p}\\
	&\quad + \E*{\abs*{\bar\pi -  \pi\p*{\Yantc{\ell-1}} }^p},\\
	\end{aligned}
	\]
	where \(\E{\abs{\bar\pi\p*{Y_\ell} -  \pi\p*{\Yantc{\ell-1}} }^p} \le b_52^{-q\p{1+p^{-1}}\ell p /2\p{q+1}}\) by \Cref{lem:ant_mil_hinge} and, for
	\[
	\Delta_\ell^{\p{\textnormal{ant}}}\Delta_k X\p*{Y_\ell}  \defeq \frac{1}{2}\sum_{i=0}^1 \Delta_k X\p*{\Yantf{\ell}{i}} - \Delta_k X\p*{\Yantc{\ell-1}},
	\]
	we have (following similar steps to \Cref{lem:2stageanty})
	\[
	\begin{aligned}
	&\E*{\abs*{\antEfmlun\p*{Y_\ell} - \antEfml\p*{\Yantc{\ell-1}} }^p}\\
	&\le \E*{\abs*{\sum_{k=0}^\ell \Mlk^{-1}\sum_{n=1}^{\Mlk}\Delta_\ell^{\p{\textnormal{ant}}}\Delta_k X^{\p{n}}\p*{Y_\ell}   }^p}\\
	&\le c_q\abs*{\sum_{k=0}^\ell \Mlk^{-1} \E*{\abs*{\Delta_\ell^{\p{\textnormal{ant}}}\Delta_k X\p*{Y_\ell}}^p }^{2/p}}^{p/2}  + 2^{p-1}\E*{\abs*{\frac{1}{2}\sum_{i=0}^1X_\ell\p{\Yantf{\ell}{i}} - X_\ell\p{\Yantc{\ell-1}}}^p}\\
	&\le c_q 2^{-\bqy[p] \ell p / 2}
	\begin{cases}
	1 & \zeta = 1\\
	\ell^{p/2} & \zeta > 1
	\end{cases}.
	\end{aligned}	
	\]
\end{proof}
From the above result, \Cref{cor:two_stage_complexity} can be extended naturally to the doubly antithetic difference \(\dlantpath\), where \Cref{assumpt:ant_Y_path,assumpt:ant_sde_density} are required in place of \Cref{assumpt:twostageY}. 			
	\section{Numerical Experiment for a Two-Stage Bermudan Option} 
	\label{sec:numerics}
\subsection{Problem Setup}
Let \(W_t = \p{W_t^i}_{i=0}^d\) be a \(d+1\)-dimensional Brownian motion defined on a probability space \(\p{\Omega, \mathcal F, \mathsf P}\), where \(W_t^i\) and \(W_t^j\) are independent for \(i\neq j\). Consider a market consisting of \(d+1\) assets  \(S_t = \p{S_t^i}_{i=0}^d\) given by  
\begin{equation}
\label{eqn:sde}
\begin{aligned}
\text{d}S_t^0 &= {\sigma^0 S_t^0}\text{d}W_t^0,\\
\text{d}S_t^i &= rS_t^i\text{d}t + {\p*{\sigma^{i}S_t^i  +  S_t^0} \text{d}W_t^i}.
\end{aligned}
\end{equation}
In this context, \(r\) represents a risk-free interest rate, \(\sigma^{i}\) is the volatility of asset \(i\) for \(i\ge 1\), and \(S_t^0\) represents a systemic stochastic volatility component.
Importantly, since the diffusion of asset \(i = 1,\dots, d\) contains a mixture of \(S_t^i\) and \(S_t^0\), the commutativity condition \cite{Kloeden:1999,giles14antmilstein,cc80} which allows Milstein simulation without evaluating the L\'evy areas is not satisfied provided \(S_0^0 \neq 0\). {However, since the coefficients of the SDE are linear, the antithetic Milstein approach in \cite{giles14antmilstein} can be used to improve convergence of multilevel differences, as discussed in \Cref{sec:ant_path_2stage}.}\\

For a fixed strike price \(K > 0\) and maturity \(T > 0\), consider a financial option with discounted payoff
\[
\pi_T\p{S_T} = e^{-rT}\max\br*{0, K - \frac{1}{d}\sum_{i=1}^d S_T^i},  
\]
at maturity and which allows an early execution at time \(T/2\), with payoff \(\pi_{T/2}\p{S_{T/2}}\). {Note that this payoff can be expressed in the form \eqref{eqn:pi_ant_mil} with a single linear projection \(P_1\p{y} = K - d^{-1}\sum_{i=1}^d y^i\).} The value of this option can be expressed in the form \eqref{eqn:nestexp2} for \(X = \pi_T\p{S_T}\), \(Y = S_{T/2}\) and \(\pi\p{Y} = \pi_{T/2}\p{Y}\). Explicitly, the Milstein scheme for this method with step size \(h_\ell\) uses the updates
\[
\begin{aligned}
\hat S_{n+1}^i &= F^i\p{\hat S_n; h, \Delta W_n, A_n}, \quad \text{where}\\
F^0\p{\hat S_n; h, \Delta W_n, A_n} &= \hat S_n^0  + \sigma^0 \hat S_n^0 \Delta W_n^0 + \frac{\p{\sigma^0}^2}{2}\hat S_n^0 \p{\Delta W_n^2 - h}\\
F^i\p{\hat S_n; h, \Delta W_n, A_n} &\defeq \hat S_n^i + r\hat S_n^ih + \p*{\sigma^i \hat S_n^i + {\hat S_n^0}}\Delta W_n^i\\ 
&\quad+ {\frac{\sigma^i}{2}\p*{\sigma^i\hat S_n^i + \hat S_n^0}\p*{\p*{\Delta W_n^i}^2 - h}}\\
&\quad + {\frac{\sigma^0}{2}\hat S_n^0\p*{\Delta W_n^i\Delta W_n^0 - A^{i, 0}_n}},
\end{aligned}
\]
where \(\Delta W_n = W_{\p{n+1}h} - W_{nh}\), and \(A = \p{A_n^{i, j}}_{ij}\) are the L\'evy areas defined in \cite[Chapter 10]{Kloeden:1999}. For the context of MLMC estimation, in \cite{giles14antmilstein}, the authors drop the L\'evy areas and couple two fine path approximations \(\hat S_{\ell, n+1}^f, \hat S_{\ell, n+1}^a\) at step size \(h_\ell \propto 2^{-\ell}\) following
\begin{equation}
\label{eqn:ant_mil_fine}
\begin{aligned}
\hat S_{\ell, n+1}^{f, i} &= F^i\p*{\hat S_{\ell, n}^{f, i} ; h_\ell, \Delta W_{n}, 0}\\
\hat S_{\ell, n+1}^{a, i} &= F^i\p*{\hat S_{\ell, n}^{a, i} ; h_\ell, \Delta W_{n + (-1)^n}, 0},
\end{aligned}
\end{equation}
and a single coarse path \(\hat S_{\ell-1, n+1}^{c}\) at level \(\ell - 1\) with step-size \(h_{\ell-1}\) given by 
\begin{equation}
\label{eqn:ant_mil_coarse}
\hat S_{\ell-1, n+2}^{c, i} = F^i\p*{\hat S_{\ell-1, n}^{c, i} ; h_{\ell-1}, \Delta W_{n} + \Delta W_{n+1}, 0}.
\end{equation}
It then follows that \cite[Theorem 4.10]{giles14antmilstein}
\begin{equation}
\label{eqn:Sbound_numerics}
\E*{\max_{0\le n\le T/h_{\ell-1}}\norm*{\frac{1}{2}\p*{\hat S_{\ell, 2n}^f + \hat S_{\ell, 2n}^a} - \hat S_{\ell-1, 2n}^c}^p} \le a_6 2^{-\ell p},
\end{equation}
for all \(p<\infty\), where \(a_6>0\) is independent of \(\ell\). In particular, since \(\pi_T\p{\cdot}\) is piecewise linear, provided
\[
{\max_{t \in \br{T/2, T}}}\prob*{\abs*{K - \frac{1}{d}\sum_{i=1}^d S_t^i} \le x} \le \bar \rho x
\]
for sufficiently small \(x>0\), \Cref{lem:ant_mil_hinge} gives{
\begin{equation}
\label{eqn:pibound_numerics}
\begin{aligned}
&\max_{t \in \br{T/2,T}} \E*{\abs*{\frac{1}{2}\p*{\pi_t\p*{\hat S_{\ell, t/h_\ell}^f} + \pi_t\p*{\hat S_{\ell, t/h_\ell}^a}} - \pi_T\p*{\hat S_{\ell-1, t/h_\ell}^c}}^p}\\
	&\le a_7 2^{-q\p{1+p^{-1}}p\ell/2\p{q+1}}
\end{aligned}
\end{equation}}
for all \(p \le q <  \infty\), where \(a_7\) is independent of \(\ell\). {In particular, define \(S_{t_1,\ell}^{f}\p{t_0, y}\), \(S_{t_1,\ell}^{a}\p{t_0, y}\) and \(S_{t_1,\ell-1}^{c}\p{t_0, y}\) to be the antithetic Milstein paths constructed through \eqref{eqn:ant_mil_fine} and \eqref{eqn:ant_mil_coarse} at levels \(\ell\) and \(\ell-1\) from time \(t_0\) to \(t_1\) given the initial condition \(S_{t_0} = y\). Given \(Y = S_{T/2}\), define antithetic multilevel differences for \(X = \pi\p{S_T}\) through
\[
\begin{aligned}
&\dl[k]X\p{Y}\defeq 
 \frac{1}{2}\p*{\pi\p*{S_{T,k}^f\p{T/2, Y}} + \pi\p*{S_{T, k}^a\p*{T/2, Y}}} - \pi\p*{S_{T,k-1}^c\p*{T/2, Y}},
\end{aligned}
\]
for \(k > 0\). By \eqref{eqn:pibound_numerics}, \Cref{assumpt:twostagegenY} holds for the differences \(\dl[k]X\) defined above with \(\bq[p] = q\p{1+p^{-1}}/\p{q+1}\) and \(\lambda_p \equiv 0\). Antithetic approximations to \(Y\) as in \Cref{sec:ant_path_2stage} are used, setting \(Y_{\ell}^{f, 0} \defeq S_{T/2, \ell}^f\p{0, S_0}\), \(Y_{\ell}^{f, 1} \defeq S_{T/2, \ell}^a\p{0, S_0}\) and \(Y_{\ell-1}^{c} \defeq S_{T/2, \ell-1}^c\p{0, S_0}\). By \eqref{eqn:Sbound_numerics} and \eqref{eqn:pibound_numerics}}, \Cref{assumpt:ant_Y_path} holds with \(\bqy[p] = q\p{1+p^{-1}}/\p{q+1}\) for all \(p\le q<\infty\). A multilevel estimate of  \eqref{eqn:nestexp2} can then be constructed of the form \eqref{eqn:mlmc} using the doubly antithetic correction terms \(\dl = \dlantpath\) analysed in \Cref{sec:ant_path_2stage}. Since \(\bq = 1\), the number of inner samples is chosen to be \(\Mlk = \Mlk[0][0]2^{\ell - k}\) (\(\zeta = 1\) in \eqref{eqn:appE}). This choice gives 
\[
\text{Cost}\p*{\dlantpath} \le \hat b_0 \p{\ell+1} 2^\ell
\]
and by \Cref{prop:moments_antithetic_path_twostage} (taking \(q\to\infty\))
\[
\var*{\dlantpath} \le \hat b_1 2^{-3\ell / 2 + \delta\ell},
\]
for all \(\delta > 0\).
The proof of \Cref{cor:two_stage_complexity} then implies the cost of attaining root mean square error \(\tol\) is
\begin{equation}
\label{eqn:mlmccost_num}
\text{Cost}_\text{MLMC}\p{\tol} \propto \tol^{-2}\p[\bigg]{\sum_{\ell = 0}^{\ceil{\abs{\log\tol}}}
	\sqrt{\ell+1}\ 2^{-\ell/4}}^2.
\end{equation}
Since the sum is bounded as \(\tol\to0\), the  asymptotic cost is of order \(\tol^{-2}\). However, for sufficiently large error tolerances, \(\p{\sum_{\ell = 0}^{\ceil{\abs{\log\tol}}} \sqrt{\ell+1}\ 2^{-\ell/4}}^2\) is roughly proportional to \(\abs{\log\tol}\), hence one can expect a logarithmic factor in the cost during a pre-asymptotic phase. \\

We further consider Euler-Maruyama simulation of \eqref{eqn:sde}, which does not require the evaluation of L\'evy areas or an antithetic path difference. Instead, we can construct an MLMC estimate for \eqref{eqn:nestexp2} using the correction terms \(\dl = \dlantmly\) considered in \Cref{sec:twostageapprox}. Therein, we take \(\dl[k]X = X_k - X_{k-1}\), where \(X_k\p{Y}\) uses Euler-Maruyama simulation of \(X = \pi_T\p{S_T}\) with step-size \(h_k\propto 2^{-k}\) given \(Y = S_{T/2}\). The intermediate state \(S_{T/2} = Y\approx Y_\ell\) is further approximated using Euler-Maruyama with step size \(h_\ell\propto 2^{-\ell}\). Under this setup, for the SDE \eqref{eqn:sde}, it follows that \Cref{assumpt:twostageY,assumpt:twostagegenY} hold for \(\bq[p] = \bqy[p] = 1\) and \(\lambda_p\equiv 0\) for all \(p\le q<\infty\). In particular, taking \(\Mlk = \Mlk[0][0]2^{\ell-k}\) in \(\dlantmly\), it follows that \(\text{Cost}\p*{\dlantmly}\le \hat b_0 \ell 2^\ell\) and \(\var{\dlantmly}\le \hat b_12^{-\ell}\). From \Cref{cor:two_stage_complexity}, the cost is of order \(\tol^{-2}\abs{\log\tol}^3\) to attain root mean square error \(\tol\). {For comparison,} by approximating \eqref{eqn:nestexp2}  using standard (single-level) Monte Carlo methods, with inner Monte Carlo simulation of \(\nestE[1]\) and Euler-Maruyama approximation of \(X\) and \(Y\), {the cost increases to order \(\tol^{-4}\) \cite{GilesHajiAliSpence:CVA,Gordy:2010}.}

\subsection{Results}
For the numerical experiment, we consider the SDE \eqref{eqn:sde} with \(d = 4,\ r = 0.05,\ S_0^0 = 0.2,\ T = 2\). The volatility of \(S_t^0\) is set as \(\sigma^0 = \sqrt{\log\p{2}/T}\) which gives \(\var{S_T^0/S_0^0}  = 1\). Moreover, we fix separate initial values \(S_0^i \in \p{9.4, 10.6}\) and volatilities \(\sigma^i\in\p{0.3, 0.4}\) for each asset \(i=1,\dots, d\). The strike price is computed as \(K = d^{-1}\sum_{i=1}^dS_0^i\).\\

\Cref{fig:twostage_bermuda_stats} plots estimated statistics of the multilevel difference \(\dl\) for each method over the range \(0\le\ell\le 10\). Each statistic is estimated using \(10^5\) independent samples. The top-left plot shows the cost of sampling each term, computed as the number of independent \(\mathcal N\p{0, 1}\) random variables required. By construction, the cost of both Euler-Maruyama and antithetic Milstein with inner samples \(\Mlk\propto 2^{\ell - k}\) is of order \(\ell2^\ell\), which aligns with \eqref{eqn:dlmly_cost}. The top-right plot shows the multilevel correction variance against \(\ell\). The observed results are consistent with the theoretical rates \(\var{\dlantmly} = \Order{2^{-\ell}}\) and \(\var{\dlantpath} = \Order{2^{-3\ell/2}}\). At level \(\ell=10\), the doubly antithetic Milstein estimate has a variance around 100 times smaller than the corresponding Euler-Maruyama approach. The bottom-left plot shows the bias of each correction statistic, which is important to determine the total number of levels \(L\) required for the MLMC estimate of \(\nestE[0]\) \eqref{eqn:mlmc} to reduce the bias sufficiently. Both methods have bias decreasing at rate \(2^{-\ell}\), although Euler-Maruyama has a pre-asymptotic phase which increases the bias by a constant factor over the Milstein difference. The bottom-right plot shows the kurtosis of \(\dl\), defined as \(\E{\abs{\dl}^4}/\var{\dl}^2\). This statistic determines the number of samples required to accurately estimate \(\var{\dl}\), which is important in estimating the optimal number of outer samples \(\{M_\ell\}_{\ell=0}^L\) required in \eqref{eqn:mlmc}. From \Cref{prop:moments_antithetic_path_twostage,lem:2stageanty}, it follows for the values of \(\bq[p],\bqy[p]\) stated above that the kurtosis is of order \(2^{\ell/2}\) for \(\dlantpath\) and of order \(1\) for \(\dlantmly\). A high kurtosis at large levels can lead to inaccurate parameter estimates when few samples are available. Since MLMC estimators aim to require a relatively small number of samples at large levels, it may not be enough to estimate \(\var{\dl}\) using the sample variance and retain efficient uncertainty quantification. One alternative is to use estimates from previous levels with sufficient sample sizes to infer the variance of higher-level differences. In \cite{Collier:CMLMC,GilesHajiAliSpence:CVA,Elfverson:2016selectiverefinement}, Bayesian extensions of this idea are considered for similar problems. \\

\Cref{fig:twostage_bermuda_mlmc} shows the cost of a single MLMC estimate against the (normalised) root mean square error \(\tol\). The cost is multiplied by \(\tol^{2}\) so that horizontal lines represent order \(\tol^{-2}\) cost. The Euler-Maruyama method is around 10 times more costly for \(\tol\approx 10^{-4}\) owing to the \(\abs{\log}^3\) factor in its cost. Note that the antithetic Milstein approach has a slight (but less prominent) increase over flat \(\tol^{-2}\) convergence. A reference line shows the extra factor is close to \(\abs{\log\tol}\), which implies the estimator is computed within a pre-asymptotic phase, as discussed above following \eqref{eqn:mlmccost_num}. An additional reference line illustrates the convergence rate for single-level Monte Carlo using Euler-Maruyama simulation. Both multilevel estimates show significant improvements over the single-level approach. The multilevel framework will lower the cost by several orders of magnitude at reasonable error tolerances.\\

Regression on the data from \Cref{fig:twostage_bermuda_stats}, accounting for pre-asymptotic behaviour, is used to select the truncation level \(L\), reducing the bias to the specified tolerance. The number of outer samples \(\br{M_\ell}_{\ell=0}^L\) required to meet the specified statistical error are estimated by \(M_\ell\approx \widehat M_\ell\) using the variance of samples at each level available to the estimator. In practice, \(1.65\widehat M_\ell\) samples are used for each level, which gives \(\approx 90\%\) confidence in the statistical error of the estimator. To illustrate this criterion, \Cref{fig:twostage_bermuda_mlmc} (Right) shows the absolute error of 10 independent MLMC averages for each choice of the multilevel differences computed at a range of tolerances. To compute the error, an estimate with normalised error \(\approx 3\times10^{-5}\) is selected to approximate the true solution. A linear line indicates the error threshold at each tolerance, accounting for the normalisation. Points below this line have an absolute error less than the target root mean square error. A clear majority of samples meet the threshold. A few outliers slightly exceed the error tolerance, which aligns with a \(10\%\) chance of observing a more significant error.
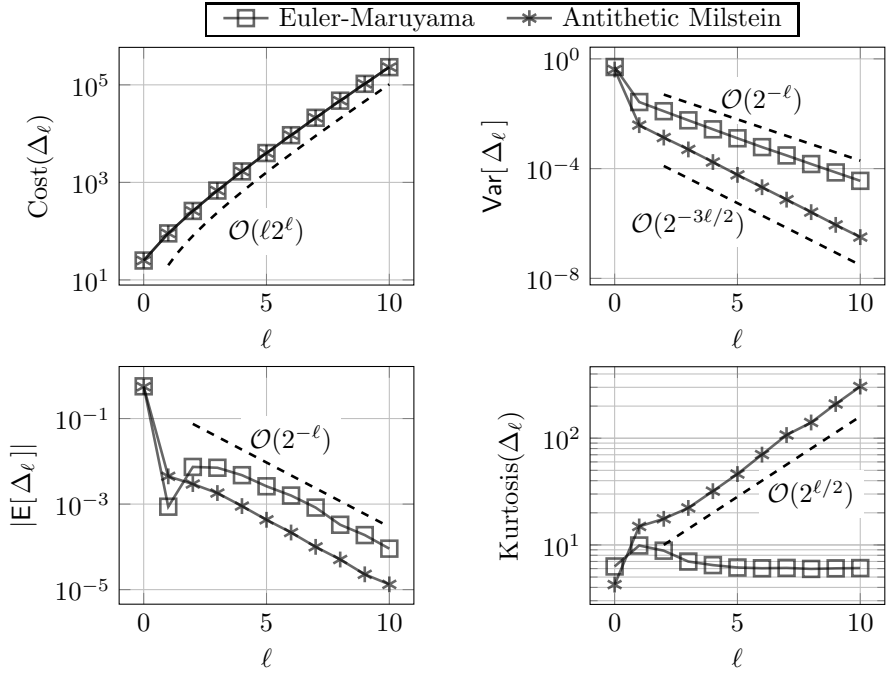
\begin{figure}
	\centering
	\begin{tabular}{|c c|}
		\hline
		\ref*{pl:em} Euler-Maruyama & \ref*{pl:mil} Antithetic Milstein \\
		\hline
	\end{tabular}
	\begin{tabular}{c p{1.5cm} c}
	\begin{tikzpicture}[trim axis left, trim axis right]
	\begin{axis}[
	xlabel = \(\ell\),
	ylabel = \(\textnormal{Cost}\p{\dl}\),
	grid = both,
	ymode = log,
	]
	\draw(5, 100) node[fill = white,]{\(\Order{\ell2^{\ell}}\)};
	\addplot+[mark = square, mark size = 3pt, black, opacity = \opac,] 
	table[x expr = \thisrow{levels}, y expr = \thisrow{sample_cost}, col sep = comma,]{data/em_T_2/level_statistics.csv};
	\label{pl:em}
	
	\addplot+[mark = asterisk, mark size = 3pt, black, opacity = \opac,] 
	table[x expr = \thisrow{levels}, y expr = \thisrow{sample_cost}, col sep = comma,]{data/mil_T_2/level_statistics.csv};
	\label{pl:mil}

	\addplot+[mark = none, dashed, black, domain = 1:10,]{10*x*2^(x)};
	
	\end{axis}
	\end{tikzpicture}
	& &
	\begin{tikzpicture}[trim axis left, trim axis right]
	\begin{axis}[
	xlabel = \(\ell\),
	ylabel = \(\var*{\dl}\),
	grid = both,
	ymode = log,
	]
	\addplot+[mark = square, mark size = 3pt, black, opacity = \opac,] 
	table[x expr = \thisrow{levels}, y expr = \thisrow{var}, col sep = comma,]{data/em_T_2/level_statistics.csv};
	
	\addplot+[mark = asterisk, mark size = 3pt, black, opacity = \opac,] 
	table[x expr = \thisrow{levels}, y expr = \thisrow{var}, col sep = comma,]{data/mil_T_2/level_statistics.csv};
	
	\draw(3, 1/10^6) node[fill = white,]{{\(\Order{2^{-3\ell/2}}\)}};
	\draw(6, 1/10^1.5) node[fill = white,]{\(\Order{2^{-\ell}}\)};
	\addplot+[mark = none, dashed, black, domain = 2:10,]{2^(-3*x/2)/1000};
	\addplot+[mark = none, dashed, black, domain = 2:10,]{2^(-2*x/2)/5};
	\end{axis}
	\end{tikzpicture}
	\\
	\begin{tikzpicture}[trim axis left, trim axis right]
	\begin{axis}[
	xlabel = \(\ell\),
	ylabel = \(\abs*{\E{\dl}}\),
	grid = both,
	ymode = log,
	]	
	\addplot+[mark = square, mark size = 3pt, black, opacity = \opac,] 
	table[x expr = \thisrow{levels}, y expr = abs(\thisrow{mean}), col sep = comma,]{data/em_T_2/level_statistics.csv};
	
	\addplot+[mark = asterisk, mark size = 3pt, black, opacity = \opac,] 
	table[x expr = \thisrow{levels}, y expr = abs(\thisrow{mean}), col sep = comma,]{data/mil_T_2/level_statistics.csv};
	\draw(6, 1/10^1.5) node[fill = white,]{\(\Order{2^{-\ell}}\)};
\addplot+[mark = none, dashed, black, domain = 2:10,]{1.5*2^(-2*x/2)/5};
	\end{axis}
	\end{tikzpicture}
	& &
	\begin{tikzpicture}[trim axis left, trim axis right]
	\begin{axis}[
	xlabel = \(\ell\),
	ylabel = \(\text{Kurtosis}\p{\dl}\),
	grid = both,
	ymode = log,
	]	
	\addplot+[mark = square, mark size = 3pt, black, opacity = \opac,] 
	table[x expr = \thisrow{levels}, y expr = abs(\thisrow{kurtosis}), col sep = comma,]{data/em_T_2/level_statistics.csv};
	
	\addplot+[mark = asterisk, mark size = 3pt, black, opacity = \opac,] 
	table[x expr = \thisrow{levels}, y expr = abs(\thisrow{kurtosis}), col sep = comma,]{data/mil_T_2/level_statistics.csv};
	
	\draw(8, 30) node[fill = white,]{\(\Order{2^{\ell/2}}\)};
	\addplot+[mark = none, dashed, black, domain = 2:10,]{5*2^(x/2)};
\end{axis}
	\end{tikzpicture}
	\end{tabular}
	\caption{Multilevel correction statistics for the two-stage Bermudan option priced using \eqref{eqn:mlmc} with \(\dl = \dlantmly\) using Euler-Maruyama path simulation and \(\dl = \dlantpath\) using antithetic Milstein path simulation.}
	\label{fig:twostage_bermuda_stats}
\end{figure} \begin{figure}
	\def\normfactor{0.5464602511221934} \centering
	\begin{tabular}{|c c|}
		\hline
		\ref*{pl:em} Euler-Maruyama &  \ref*{pl:mil} Antithetic Milstein\\
		 \ref*{pl:logtol} Order \(\tol^{-2}\abs{\log\tol}\) &  \ref*{pl:logtol3} Order \(\tol^{-2}\abs{\log\tol}^3\) \\
		 \multicolumn{2}{|c|}{\ref*{pl:tolm4} Order \(\tol^{-4}\) (Standard Monte Carlo)}   \\
		  \hline
	\end{tabular}
	\begin{tabular}{c p{1.5cm} c}
	\begin{tikzpicture}[trim axis left, trim axis right]
		\begin{axis}[
		xlabel = Normalised \(\tol\),
		ylabel = \(\textnormal{Cost}_{\textnormal{MLMC}}\times {\tol}^{2}\),
		grid = both,
		ymode = log,
		xmode = log,
		ymax = 10^5,
]
		\addplot+[mark = none, densely dotted, black, domain = 2*10^(-4):1.6*10^(-2),opacity = \opac,]{0.06*x^(-2)};
		\label{pl:tolm4}
		\addplot+[mark = none, dashed, black, domain = 2*10^(-5):2*10^(-2),opacity = \opac,]{8*(abs(ln(x*\normfactor/2.33)))^3};
		\label{pl:logtol3}		
,		\addplot+[mark = none, black, domain = 2*10^(-5):2*10^(-2.5),opacity = \opac,]{40*(abs(ln(x*\normfactor)))};
		\label{pl:logtol}	
			
		\addplot+[mark = square, mark size = 3pt, black, opacity = \opac,] 
		table[x expr = \thisrow{tol}/\normfactor, y expr = \thisrow{cost}*(\thisrow{tol}/\normfactor)^2, col sep = comma, skip coords between index={0}{2}]{data/em_T_2/mlmc_statistics.csv};
		
		\addplot+[mark = asterisk, mark size = 3pt, black, opacity = \opac,skip coords between index={0}{2}] 
		table[x expr = \thisrow{tol}/\normfactor, y expr = \thisrow{cost}*(\thisrow{tol}/\normfactor)^2, col sep = comma,]{data/mil_T_2/mlmc_statistics.csv};

		\end{axis}
	\end{tikzpicture}
	& &
	\begin{tikzpicture}[trim axis left, trim axis right]
	\begin{axis}[
	xlabel = Normalized \(\tol\),
	ylabel = Absolute Error,
xmode = log,
	ymode = log,
	]
	\addplot+[only marks, mark = square, mark size = 3pt, black, opacity = 0.3] 
	table[x expr = \thisrow{tol}/\normfactor, y expr = abs(\thisrow{P} - \normfactor), col sep = comma,]{data/em_T_2/Ps.csv};

	\addplot+[only marks, mark = asterisk, mark size = 3pt, black, opacity = 0.3] 
	table[x expr = \thisrow{tol}/\normfactor, y expr = abs(\thisrow{P} - \normfactor), col sep = comma,]{data/mil_T_2/Ps.csv};

	\addplot+[mark = none, black, domain = (1/2^12.3)/\normfactor:1/2^6.7/\normfactor,]{x*\normfactor};

	\end{axis}
	\end{tikzpicture}
	\end{tabular}
	\caption{(Left) MLMC cost times \(\tol^2\) against root mean square error \(\tol\) (normalised according to the approximate true solution) for the two-stage Bermudan option. (Right) Absolute error against Normalized \(\tol\) for 10 independent MLMC runs. Results are shown for the correction terms \(\dl = \dlantmly\) using Euler-Maruyama and \(\dl = \dlantpath\) using antithetic Milstein.}
	\label{fig:twostage_bermuda_mlmc}
\end{figure}
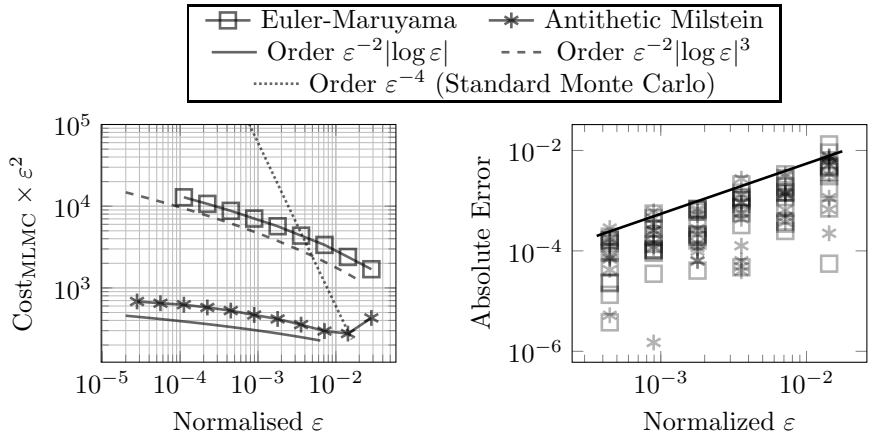

	\section{Conclusion}
We considered the use of recursive multilevel Monte Carlo techniques to approximate the quantity \(\nestE[0]\) defined through \eqref{eqn:nestexp2}. In \Cref{sec:ant_approx}, an antithetic nested multilevel estimator is constructed, which allows for biased approximations to \(\nestE[1]\p{Y} = \E{X\given Y}\) in the case where \(X\) cannot be sampled exactly. \Cref{sec:twostageapprox} extends this approach to include approximate simulation of \(Y\). Under general conditions, the resulting estimators of \(\nestE[0]\) are shown to obtain root mean square error \(\tol\) with cost of order \(\tol^{-2}\) or \(\tol^{-2}\abs{\log\tol}^3\), depending on the convergence of the approximations to \(X\) and \(Y\). In \Cref{sec:ant_approx}, a doubly antithetic estimator is proposed, which can obtain order \(\tol^{-2}\) asymptotic cost for a broader class of applications.\\

	One future direction of interest involves generalising \eqref{eqn:nestexp2} to include \(\nestE[0] = \E{f\p{\nestE[1]\p{Y}, Y}}\) for more general piecewise-smooth functions \(f\) and allowing \(X\) to take values in \(\rset^d\). Much of the analysis within this paper will remain the same, using \Cref{lem:ant_general_smooth} in place of \Cref{lem:ant_general}. Another application of interest is the extension to arbitrary depths of nested expectations, and a comparison with the recursive randomised MLMC approach in \cite{zhou2022,syed2023optimal}. By recursively using the methods of \Cref{sec:twostageapprox,sec:ant_path_2stage},  one could also consider the approximation of repeatedly nested expectations with approximate samples of \(Y\). This approach has potentially interesting applications for Bermudan option pricing where the underlying market is approximated either by Euler-Maruyama or Milstein discretisation.
	
	\subsubsection*{Acknowledgements}
	JS was supported by the EPSRC Centre for Doctoral Training in Mathematical Modelling, Analysis and Computation (MAC-MIGS) funded by the UK Engineering and Physical Sciences Research Council (grant EP/S023291/1), Heriot-Watt University and the University of Edinburgh.\\
	The authors are thankful for helpful discussions with Professor Michael B. Giles at an early stage of this project, in particular with regards to \Cref{lem:ant_general_smooth} for the smooth function case.
	
	\bibliographystyle{abbrv}
	\bibliography{../references}
	
	\appendix 
	\section{General Antithetic Result}
		This appendix proves a stronger version of \Cref{lem:ant_general}, considering more general functions of the random variable of interest. The extended result is useful in \Cref{sec:ant_path_2stage} when proving convergence of antithetic correction terms for \(\pi\p{Y}\).

\begin{lemma}
	\label{lem:ant_general_smooth}
	Consider a function \(g:\rset^d\to\rset\) mapping \(z\in\rset^d\) to
	\[
		g\p{z} = g_s\p{z} + \sum_{r=1}^R \max\br*{P_r\p{z},0},
	\]
	for fixed \(R\in\zset_{\ge 0}\), where 
	\begin{itemize}
		\item \(g_s:\rset^d\mapsto\rset\) is twice differentiable with second derivatives bounded by \(L_s^{\prime\prime}\ge 0\) and \(\alpha\)-H\"older continuous such that
		\[
			\abs{g_s\p{z_0} - g_s\p{z_1}}\le L_s\norm{z_0-z_1}^\alpha,
		\]
		for all \(z_0, z_1\).
		\item Each \(P_r:\rset^d\mapsto\rset\) is an affine projection of \(z\) into \(\rset\).
	\end{itemize}
	Let \(Z\) represent an \(\rset\)-valued random variable such that for some \(\bar \delta, \bar \rho > 0\),  we have
	\begin{equation}
		\label{eqn:cdf_bound_general_smooth}
		\sup_{1\le r\le R}\prob{\abs{P_r\p{Z}} < x} \le \bar\rho x,	
	\end{equation}
	for all \(x \le \bar\delta\). Consider a sequence of approximations \(\{Z_\ell\}_{\ell\in\zset_{\ge 0}}\) to \(Z\) satisfying the condition that for some \(q > 2\) we have
	\[
		\sup_{\ell\ge0}\E*{\norm{Z_\ell - Z}^q} < \infty.
	\]
	Let	\(Z_\ell^{\p{0}}, Z_\ell^{\p{1}}\) be two samples of \(Z_\ell\). For \(2\le p\le q/\alpha\),
	there is \(b_0 > 0\), independent of \(\ell\), such that
	\[
		\begin{aligned}
			&\E*{\abs*{ g\p*{\frac{Z_\ell^{\p{0}} + Z_\ell^{\p{1}}}{2}} - \frac{1}{2}\sum_{i=0}^1 g\p*{Z_\ell^{\p{i}}}}^p} \\
			&\le b_0\p*{
											\p{L_s^{\prime\prime}}^q\E*{\norm*{Z_\ell - Z}^q}^{\min\br{2p/q, q}} 
											+ R^q\E*{\norm*{Z_\ell - Z}^q}^{\p{p+1}/\p{q+1}}
										}.
		\end{aligned}
	\]
\end{lemma}
\begin{remark}
{Similar results are proven in \cite[Theorem 4.1]{GilesHajiAli:2018} and \cite[Theorem 3]{gg19} for \(\rset\)-valued variables \(Z\) and \(g\p{z} = \max\br{z, 0}\). The case  where \(g\) is piecewise linear and for a specific random variable \(Z \in\p{0,1}\) is considered in \cite{BujokK2015}. The authors in \cite[Proposition 5]{Bourgey20} and \cite[Theorem 5.2]{giles14antmilstein} consider a similar class of functions \(g\) to \Cref{lem:ant_general_smooth}. Since the complexity of MLMC sums requires convergence of the second moment, many results are stated for \(p = 2\) in \Cref{lem:ant_general_smooth}. Since many extensions of MLMC methods such as adaptivity \cite{GilesHajiAli:2018,hajiali2021adaptive} or branching \cite{gh22b} involve convergence of higher moments, the proof below aims to capture convergence of general moments \(p\le q\).}
\end{remark}
\begin{proof}
	Let 
	\[
		\bar Z_\ell \defeq \frac{Z_\ell^{\p{0}} + Z_\ell^{\p{1}}}{2}.
	\]
	By Jensen's inequality, note that 
	\[
		\begin{aligned}
			&\E*{\abs*{ g\p*{\frac{Z_\ell^{\p{0}} + Z_\ell^{\p{1}}}{2}} - \frac{1}{2}\sum_{i=0}^1 g\p*{Z_\ell^{\p{i}}}}^p}\\
			 &\le \p{R+1}^{p-1}\p*{\E*{\abs*{\Delta_{s, \ell}}^p} + \sum_{r=1}^R d_r^p\E*{\abs*{\Delta_{r, \ell}}^p}},
		\end{aligned}
	\]
	where 
	\[
		\begin{aligned}
				\Delta_{s, \ell} &= g_s\p*{\bar Z_\ell} - \frac{1}{2}\sum_{i=0}^1 g_s\p*{Z_\ell^{\p{i}}}\\
				\Delta_{r, \ell} &= \max\br*{\bar Z_\ell, e_r} - \frac{1}{2}\sum_{i=0}^1 \max\br*{Z_\ell^{\p{i}}, e_r}.
		\end{aligned}
	\]
	For the first term, using the H\"older continuity condition and combining second order Taylor expansions gives
	\[
		\begin{aligned}
			\abs*{\Delta_{s, \ell}} &\le L_s\norm*{{Z_\ell^{\p{0}}} - {Z_\ell^{\p{1}}}}^\alpha\\
			\abs*{\Delta_{s, \ell}} &\le \frac{L_s^{\prime\prime}}{2}\norm*{{Z_\ell^{\p{0}}} - {Z_\ell^{\p{1}}}}^2.
		\end{aligned}
	\]
	By taking appropriate powers of the above expansions, with \(v = \max\br{0, \p{2p-q}/p\p{2-\alpha}}\), for \(p \le q/\alpha\) it follows that
	\[
		\begin{aligned} 
			\E*{\abs*{\Delta_{s, \ell}}^p} &\le \E*{\abs*{\Delta_{s, \ell}}^{\alpha vp}\abs*{\Delta_{s, \ell}}^{2\p{1-v}p}}\\
				&\le 2^{q-1}L_s^{vp}\p{L_s^{\prime\prime}}^{\p{1-v}p}\E*{\norm*{Z_\ell - Z}^{\min\br{2p, q}}}\\
				&\le 2^{q-1}L_s^{vp}\p{L_s^{\prime\prime}}^{\p{1-v}p}\E*{\norm*{Z_\ell - Z}^q}^{\min\br{2p/q, q}},
		\end{aligned}
	\]
	using H\"older's inequality in the final line.\\
	
For the second term, note that since \(P_r\) is affine,  \(\Delta_{r, \ell}\) is non-zero only when the indicator
	\[
		I \defeq \I{P_r\p{Z_\ell^{\p{0}}}P_r\p{Z_\ell^{\p{1}}} < 0}
	\]
	is equal to 1. The event \(I = 1\) implies either \(P_r\p{Z_\ell^{\p{0}}} P_r\p{Z} < 0\) or \(P_r\p{Z_{\ell}^{\p{1}}}P_r\p{Z} < 0\). Hence, we can bound
	\[
		I \le I_0 + I_1,
	\]
	where we define
	\[
		\begin{aligned}
			I_i &\defeq \I{\abs{P_r\p{Z}} < \abs{P_r\p{Z^{\p{i}}} - P_r\p{Z}}},\\
			\hat I_i &\defeq \I{\abs{P_r\p{Z^{\p{i}}} - P_r\p{Z}} < \abs{P_r\p{Z^{\p{0}}} - P_r\p{Z^{\p{1}}}}}.
		\end{aligned}
	\]
	Then, 
	\[
		\begin{aligned}
			\E*{\abs*{\Delta_{r,\ell}}^p} &= \E*{\abs*{\Delta_{r,\ell}}^pI}\\
			&\le \E*{\abs*{\Delta_{r,\ell}}^p\p*{I_0 + I_1}}.
		\end{aligned}
	\]
	For \(i = 0, 1\) and any \(0<\psi\le\bar\delta\), letting \(L_r\) denote the Lipschitz constant of \(P_r\), it follows from \eqref{eqn:cdf_bound_general} that 
	\[
		\begin{aligned}
			&\E*{\abs*{\Delta_{r,\ell}}^pI_i}\\
			 &\le  L_r^p\E*{\abs*{P_r\p{Z_{\ell}^{\p{0}}} - P_r\p{Z_\ell^{\p{1}}}}^p\p*{I_i\hat I_i + I_i\p{1 - \hat I_i}}}\\
			&\le  L_r^p\E*{\abs*{P_r\p{Z_{\ell}^{\p{0}}} - P_r\p{Z_\ell^{\p{1}}}}^p \I{\abs{P_r\p{Z}} < \abs{P_r\p{Z_\ell^{\p{0}}} - P_r\p{Z_\ell^{\p{1}}}}}}\\
			&\quad +    L_r^p\E*{\abs*{P_r\p{Z_{\ell}^{\p{i}}} - P_r\p{Z}}^p \I{\abs{P_r\p{Z}} < \abs{P_r\p{Z_\ell^{\p{i}}} - P_r\p{Z}}}}\\
			&\le  L_r^p\E*{\abs*{P_r\p{Z_{\ell}^{\p{0}}} - P_r\p{Z_\ell^{\p{1}}}}^p \p*{\I{\abs{P_r\p{Z}} < \abs{P_r\p{Z_\ell^{\p{0}}} - P_r\p{Z_\ell^{\p{1}}}} < \psi} + \I{\psi < \abs{P_r\p{Z_\ell^{\p{0}}} - P_r\p{Z_\ell^{\p{1}}}}}} }\\
			&\quad +    L_r^p\E*{\abs*{P_r\p{Z_{\ell}^{\p{i}}} - P_r\p{Z}}^p \p*{\I{\abs{P_r\p{Z}} <  \abs{P_r\p{Z_\ell^{\p{i}}} - P_r\p{Z}} < \psi} + \I{\psi < \abs{P_r\p{Z_\ell^{\p{i}}} - P_r\p{Z}}}}}\\
			&\le  2L_r^p\bar\rho\psi^{p+1} +  L_r^q\p{1+2^q}\psi^{p-q} \E*{\norm*{Z_{\ell}- Z}^q} .
		\end{aligned}
	\]
	Hence,
	\[
	\E{\abs*{\Delta_{r,\ell}}^p} \le 4L_r^p\bar\rho \psi^{p+1} + 2\p{1+2^q}L_r^q\psi^{p-q}\E*{\norm*{Z_\ell - Z}^q}.
	\]
	Taking 
	\[
		\psi = \min\br*{1, \frac{\bar\delta}{\sup_{\ell\ge 0}\E*{\norm{Z_\ell - Z}^q}}}\p*{\E*{\norm*{Z_\ell - Z}^q}}^{1/\p{q+1}}
	\]
	completes the result.
\end{proof} 
\end{document}